\documentclass[a4paper,reqno]{amsart}
\usepackage{amsmath,amssymb,mathrsfs}
\usepackage{txfonts,bm,pifont}
\usepackage{cite}
\usepackage{float}
\usepackage{enumitem}
\usepackage{hyperref}
\usepackage{graphicx}

\newtheorem{theorem}{Theorem}[section]
\newtheorem{proposition}[theorem]{Proposition}

\newtheorem{lemma}[theorem]{Lemma}
\newtheorem{remark}[theorem]{Remark}

\numberwithin{equation}{section}
\numberwithin{figure}{section}
\newcommand{\wutilde}[1]{\vrule depth 0pt width 0pt%
{\raise0.8pt\hbox{$\smash{{\mathop{#1} \limits_{\displaystyle\widetilde{}}}}$}}}

\newcommand{\al}{\alpha}

\newcommand{\de}{\delta}

\newcommand{\ep}{\bm{\epsilon}}
\newcommand{\ka}{\kappa}

\newcommand{\bbZ}{\mathbb{Z}}

\newcommand{\bbC}{\mathbb{C}}

\newcommand{\br}[1]{{\langle{#1}\rangle}}
\makeatletter
\long\def\@makecaption#1#2{
 \vskip 10pt
 \setbox\@tempboxa\hbox{#1. #2}
 \ifdim \wd\@tempboxa >\hsize #1. #2\par \else \hbox
to\hsize{\hfil\box\@tempboxa\hfil}
 \fi}
\makeatother
\newcommand{\orcidauthorA}{0000-0001-7504-4444}

\begin{document}
\title[]{Discrete Power Functions on a Hexagonal Lattice I:\\
Derivation of defining equations from the symmetry of the Garnier System in two variables}

\author{Nalini Joshi}
\thanks{NJ's ORCID ID is \orcidauthorA.}
\address{School of Mathematics and Statistics F07, The University of Sydney, NSW 2006, Australia.}
\email{nalini.joshi@sydney.edu.au}
\author{Kenji Kajiwara}
\address{Institute of Mathematics for Industry, Kyushu University, 744 Motooka, Fukuoka 819-0395, Japan.}
\email{kaji@imi.kyushu-u.ac.jp}
\author{Tetsu Masuda}
\address{Department of Mathematical Sciences, Aoyama Gakuin University, Sagamihara, Kanagawa 252-5258, Japan.}
\email{masuda@gem.aoyama.ac.jp}
\author{Nobutaka Nakazono}
\address{Institute of Engineering, Tokyo University of Agriculture and Technology, 2-24-16 Nakacho Koganei, Tokyo 184-8588, Japan.}
\email{nakazono@go.tuat.ac.jp}
\begin{abstract}
The discrete power function on the hexagonal lattice proposed by Bobenko et al is considered, whose defining equations consist of three cross-ratio equations and a similarity constraint. We show that the defining equations are derived from the discrete symmetry of the Garnier system in two variables.
\end{abstract}
\maketitle
\setcounter{tocdepth}{1}

\allowdisplaybreaks

\section{Introduction}
The theory of discrete analytic functions based on circle packings and circle patterns was initiated by Thurston's idea of using circle packings to approximate the Riemann mapping \cite{Thurston1985}. Recently various applications of circle packings and circle patterns have been reported, for example, in probability and statistical physics \cite{Smirnov:2010, Lis2019} based on general theory (see for example \cite{Stephenson2003,Stephenson:book2005} and references therein).

A theory of discrete analytic functions based on the circle patterns generated by discrete integrable systems has been developed \cite{BobenkoA2008:MR2467378}.
Nijhoff et al. \cite{NRGO2001:MR1819383} showed that the cross-ratio equation (the discrete Schwarzian KdV equation) together with a certain similarity constraint can be regarded as a part of the B\"acklund transformations of the Painlev\'e VI equation. 
Bobenko found that this system of equations defines a discrete power function on a square lattice \cite{Bobenko99}, and its various properties have been investigated in detail \cite{Agafonov:2003,Agafonov:2005,AB:2000,Bobenko-Its, BP1999:MR1676682}. 
In this context, the cross-ratio equation can be regarded the nonlinear discretization of the Cauchy-Riemann relation.
In our previous studies \cite{AHKM:circle_patterns,JKMNS2017:MR3741826}, we constructed the explicit formulas for the discrete power functions on a square lattice in terms of the hypergeometric $\tau$-functions of the Painlev\'e VI equation, and clarified the complete characterization of the defining equations in relation to the symmetry of the Painlev\'e VI equation that is isomorphic to an extension of the affine Weyl group of type $D_4^{(1)}$. 

Another type of discrete power function defined on a hexagonal lattice has been investigated in \cite{AB2003:MR2006759,BobenkoHoffmann,BobenkoHoffmann03}. 
Consider the following system of difference equations consisting of three cross-ratio equations:
\begin{equation}\label{cr:3d}
\begin{array}{l}
\dfrac{(f_{l_1,l_2,l_3}-f_{l_1+1,l_2,l_3}) (f_{l_1+1,l_2+1,l_3}-f_{l_1,l_2+1,l_3})}
 {(f_{l_1+1,l_2,l_3}-f_{l_1+1,l_2+1,l_3})(f_{l_1,l_2+1,l_3}-f_{l_1,l_2,l_3})}
 =\dfrac{1}{x_1},\\[4mm]
\dfrac{(f_{l_1,l_2,l_3}-f_{l_1,l_2+1,l_3}) (f_{l_1,l_2+1,l_3+1}-f_{l_1,l_2,l_3+1})}
 {(f_{l_1,l_2+1,l_3}-f_{l_1,l_2+1,l_3+1})(f_{l_1,l_2,l_3+1}-f_{l_1,l_2,l_3})}
 =\dfrac{1}{x_2},\\[4mm]
\dfrac{(f_{l_1,l_2,l_3}-f_{l_1,l_2,l_3+1}) (f_{l_1+1,l_2,l_3+1}-f_{l_1+1,l_2,l_3})}
 {(f_{l_1,l_2,l_3+1}-f_{l_1+1,l_2,l_3+1})(f_{l_1+1,l_2,l_3}-f_{l_1,l_2,l_3})}
 =\dfrac{1}{x_3},
\end{array}
\end{equation}
and a similarity constraint:
\begin{equation}\label{eqn_of_power_fnct}
\begin{array}{ll}
\al^0_0\, f_{l_1,l_2,l_3}\!\!\!
&=(l_1-\al^1_1)\dfrac{(f_{l_1+1,l_2,l_3}-f_{l_1,l_2,l_3})(f_{l_1,l_2,l_3}-f_{l_1-1,l_2,l_3})}{f_{l_1+1,l_2,l_3}-f_{l_1-1,l_2,l_3}}\\[4mm]
&+(l_2-\al^2_1)\dfrac{(f_{l_1,l_2+1,l_3}-f_{l_1,l_2,l_3})(f_{l_1,l_2,l_3}-f_{l_1,l_2-1,l_3})}{f_{l_1,l_2+1,l_3}-f_{l_1,l_2-1,l_3}}\\[4mm]
&+(l_3-\al^3_1)\dfrac{(f_{l_1,l_2,l_3+1}-f_{l_1,l_2,l_3})(f_{l_1,l_2,l_3}-f_{l_1,l_2,l_3-1})}{f_{l_1,l_2,l_3+1}-f_{l_1,l_2,l_3-1}}. 
\end{array}
\end{equation}
Here, $l_1,l_2,l_3$ are integers, and $x_i$, $i=1,2,3$, $\al^j_1$, $j=1,2,3$, and $\al^0_0$ are complex parameters satisfying $x_1x_2x_3=1$. 
A solution $f_{l_1,l_2,l_3}$ of equations \eqref{cr:3d} and \eqref{eqn_of_power_fnct}, which satisfies the initial conditions
\begin{equation}\label{initial_conditions}
f_{1,0,0} =1,\quad 
f_{0,1,0}=e^{c(a_2+a_3)\sqrt{-1}},\quad 
f_{0,0,-1}=e^{ca_3\sqrt{-1}},
\end{equation}
with the identification $x_j=e^{2a_j\sqrt{-1}}\,(a_j>0,\,a_1+a_2+a_3=\pi)$, $j=1,2,3$, and $\al_0^0=c/2$\, $(0<c<2)$, $\al_1^1=\al_1^2=\al_1^3=0$, is associated with a hexagonal circle pattern \cite[Definition 1]{AB2003:MR2006759}. 
For a given set of integers $(l_1,l_2,l_3)$ with $l_1+l_2+l_3=0$, the six points $f_{l_1\pm 1,l_2,l_3}$, $f_{l_1,l_2\pm 1,l_3}$ and $f_{l_1,l_2,l_3\pm 3}$ lie on a circle with the center at  $f_{l_1,l_2,l_3}$, which gives rise to a circle pattern.  
See Figure \ref{fig:hexxagonal} for an example with $a_i=\pi/3$, $i=1,2,3$, and $c=3/2$.
We refer to such a solution as the discrete power function on a hexagonal lattice. Also, equations \eqref{cr:3d} and \eqref{eqn_of_power_fnct} are referred to as its defining equations. 

\begin{figure}[H]
\begin{center}
\includegraphics[width=0.5\textwidth]{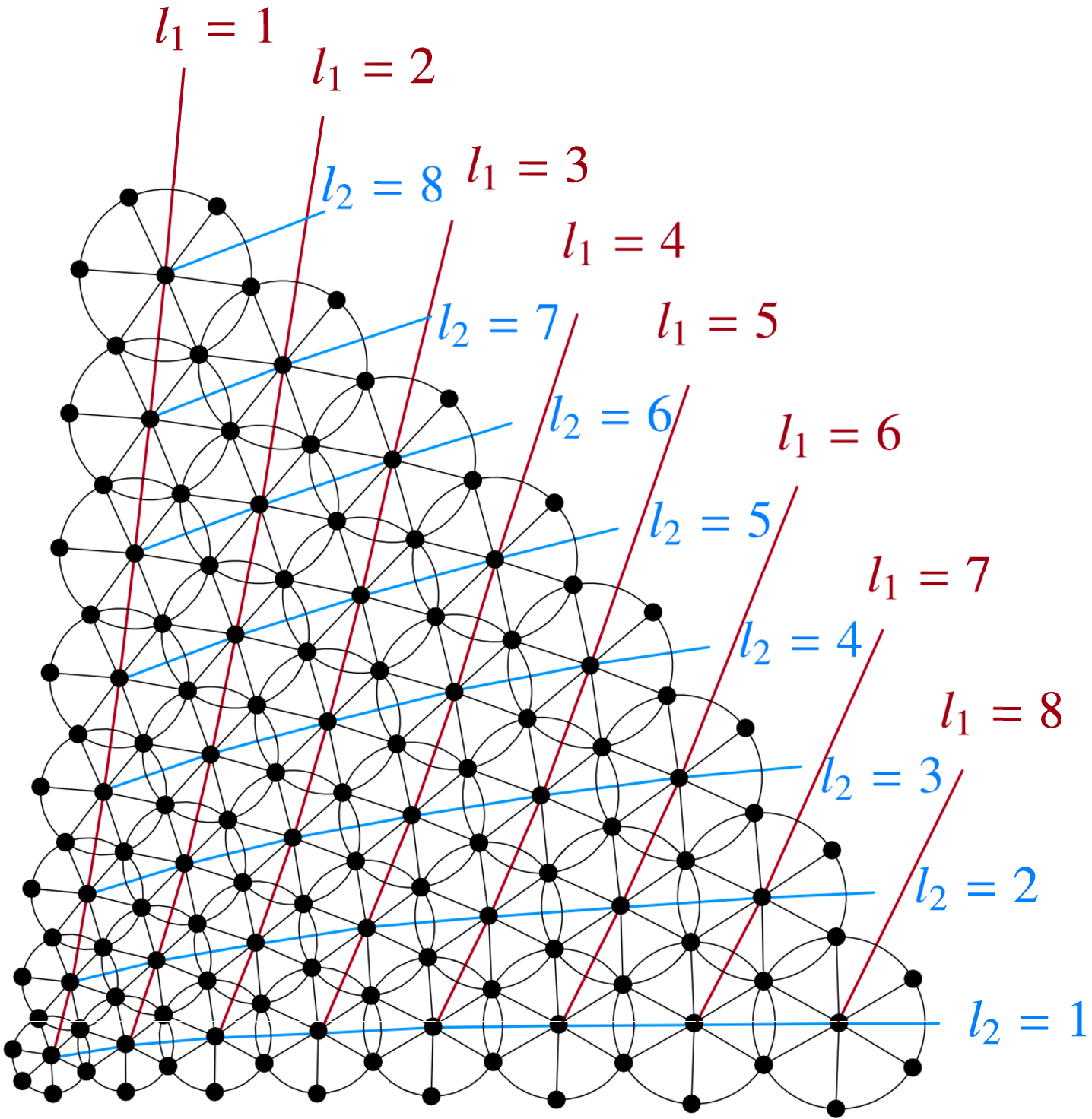}
\qquad
\includegraphics[width=0.4\textwidth,trim=5 100 50 50]{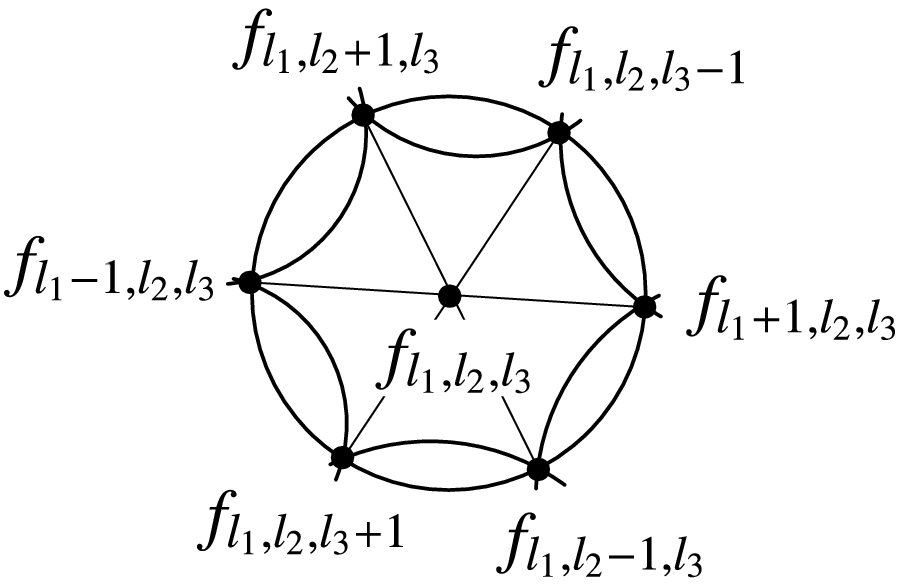}
\end{center}
\caption{A hexagonal circle pattern with $a_1=a_2=a_3=\pi/3$ and $c=3/2$. Left: circle pattern with the coordinate lines. Right: configuration of the lattice points around $f_{l_1,l_2,l_3}$ with $l_1+l_2+l_3=0$. }
\label{fig:hexxagonal}
\end{figure}

The main purpose of this paper is to show that the defining equations of the discrete power functions on a hexagonal lattice \eqref{cr:3d} and \eqref{eqn_of_power_fnct} arise from the discrete symmetry of the Garnier system in two variables (see \cite{KO1984:MR776915,SuzukiT2005:MR2177118} for details of the Garnier system). 
In the next publication, we plan to construct the explicit formulas
for the discrete power functions on the hexagonal lattice in terms of the hypergeometric $\tau$ functions of the Garnier system.

This paper is organized as follows.  
In Section \ref{section:Garnier}, we formulate the discrete symmetry of the Garnier system in two variables, and derive some bilinear relations for the so called $\tau$-variables. 
In Section \ref{Section:diffence equations}, we derive a system of partial difference equations with some constraints from the bilinear relations. As a sub-system of the difference equations, we obtain the equations \eqref{cr:3d} and \eqref{eqn_of_power_fnct}. 
Concluding remarks are then given in Section \ref{ConcludingRemarks}. 
It is known that integrable difference equations on a hexagonal lattice can be reduced from equations on a multidimensional quadrilateral lattice \cite{Doliwa et al 2007}. 
In Appendix \ref{section:ABS_Garnier}, we consider a multi-dimensionally consistent system of partial difference equations, and show that it is reduced to the system of partial difference equations discussed in Section \ref{Section:diffence equations}.

\section{Discrete symmetry of the Garnier system in two variables}\label{section:Garnier}
In this section, we first formulate the discrete symmetry of the Garnier system in two variables. Then we introduce $\tau$-variables through the Hamiltonians of the Garnier system in two variables and derive several bilinear relations among the $\tau$-variables.

\subsection{Garnier system in two variables and its symmetry}
The Garnier system in two variables\,\cite{Garnier, IKSY, KO1984:MR776915} is equivalent to the Hamiltonian system 
\begin{equation}
\dfrac{\partial q_j}{\partial t_i}=
\dfrac{\partial H_i}{\partial p_j},\quad
\dfrac{\partial p_j}{\partial t_i}=
-\dfrac{\partial H_i}{\partial q_j},\quad(i,j=1,2),
\end{equation}
with the Hamiltonians 
\begin{equation}\label{Garnier_hamiltonian}
\begin{array}{ll}
t_i(t_i-1)H_i\!\!\!\!
&=q_i(q_1p_1+q_2p_2+\al)(q_1p_1+q_2p_2+\al+\ka_{\infty})
+t_ip_i(q_ip_i-\theta_i)\\[1mm]
&-\dfrac{t_j(t_i-1)}{t_i-t_j}(q_jp_j-\theta_j)q_ip_j
-\dfrac{t_i(t_i-1)}{t_i-t_j}(q_ip_i-\theta_i)q_jp_i\\[3mm]
&-\dfrac{t_i(t_j-1)}{t_j-t_i}q_jp_j(q_ip_i-\theta_i)
-\dfrac{t_i(t_j-1)}{t_j-t_i}q_ip_i(q_jp_j-\theta_j)\\[3mm]
&-(t_i+1)(q_ip_i-\theta_i)q_ip_i+(\ka_1t_i+\ka_0-1)q_ip_i,
\quad(i,j)=(1,2),~(2,1),
\end{array}
\end{equation}
where $\al=-\dfrac{1}{2}(\theta_1+\theta_2+\ka_0+\ka_1+\ka_{\infty}-1)$. 

This system admits the following birational canonical transformations \cite{IKSY,Kimura} 
\begin{equation}
\begin{array}{lll}
\sigma_{12}&:&
\theta_1\leftrightarrow\theta_2,\quad t_1\leftrightarrow t_2,\quad
q_1\leftrightarrow q_2,\quad p_1\leftrightarrow p_2,\\[1mm]
\sigma_{23}&:&
\theta_2\leftrightarrow\ka_0,\quad
t_1\mapsto\dfrac{t_2-t_1}{t_2-1},\quad
t_2\mapsto\dfrac{t_2}{t_2-1},\\[3mm]&&
q_1\mapsto\dfrac{t_2-t_1}{t_1(t_2-1)}q_1,\quad
q_2\mapsto\dfrac{t_2}{1-t_2}(g_t-1),\\[3mm]&&
p_1\mapsto\dfrac{t_2-1}{t_2-t_1}(t_1p_1-t_2p_2),\quad
p_2\mapsto(1-t_2)p_2,\\[3mm]
\sigma_{34}&:&
\ka_0\leftrightarrow\ka_1,\quad t_i\mapsto\dfrac{1}{t_i},
\quad q_i\mapsto\dfrac{q_i}{t_i},\quad p_i\mapsto t_ip_i,\\[3mm]
\sigma_{45}&:&
\ka_1\leftrightarrow\ka_{\infty},
\quad t_i\mapsto\dfrac{t_i}{t_i-1},\\[3mm]&&
q_i\mapsto\dfrac{q_i}{g_1-1},\quad 
p_i\mapsto(g_1-1)(p_i-\al-q_1p_1-q_2p_2),
\end{array}
\end{equation}
where 
\begin{equation}\label{def_g}
g_t=\dfrac{q_1}{t_1}+\dfrac{q_2}{t_2},\quad g_1=q_1+q_2. 
\end{equation}
Note that these transformations act on the parameters $\theta=(\theta_1,\theta_2,\ka_0,\ka_1,\ka_{\infty})$ by transposing adjacent pairs, and generate the group isomorphic to the symmetric group $\mathfrak{S}_5$. 
The following mappings
\begin{equation}
\begin{array}{lll}
r_i&:&\theta_i\mapsto-\theta_i,\quad 
p_i\mapsto p_i-\dfrac{\theta_i}{q_i}\quad(i=1,2),\\[4mm]
r_3&:&\ka_0\mapsto-\ka_0,
\quad p_i\mapsto p_i-\dfrac{\ka_0}{t_i(g_t-1)},\\[4mm]
r_4&:&\ka_1\mapsto-\ka_1,
\quad p_i\mapsto p_i-\dfrac{\ka_1}{g_1-1},\\[3mm]
r_5&:&\ka_{\infty}\mapsto-\ka_{\infty},
\end{array}
\end{equation}
are also birational canonical transformations for the Garnier system in two variables. 
These satisfy the fundamental relations 
\begin{equation}
{r_i}^2=1\quad(i=1,\ldots,5),\qquad r_ir_j=r_jr_i\quad(i\ne j),\qquad
\sigma r_i\sigma^{-1}=r_{\sigma(i)}\quad(\sigma\in\mathfrak{S}_5).
\end{equation}
Further there exists the transformation\,\cite{Tsuda1}
\begin{equation}
\begin{array}{lll}
r_{34}&:&
\theta_i\mapsto-\theta_i\,(i=1,2),
\quad\ka_0\mapsto-\ka_0+1,
\quad\ka_1\mapsto-\ka_1+1,
\quad\ka_{\infty}\mapsto-\ka_{\infty},\\[1mm]
&&q_i\mapsto\dfrac{t_ip_i(q_ip_i-\theta_i)}
{(q_1p_1+q_2p_2+\al)(q_1p_1+q_2p_2+\al+\ka_{\infty})},\quad
q_ip_i\mapsto-q_ip_i.
\end{array}
\end{equation}
We have the fundamental relations 
\begin{equation}
{r_{34}}^2=1,\qquad r_ir_{34}=r_{34}r_i\quad(i\ne3,4),
\end{equation}
and 
\begin{equation}
r_{34}\sigma_{34}=\sigma_{34}r_{34},\quad
r_{34}\sigma_{ab}=\sigma_{ab}r_{34}\,(a,b\ne3,4),
\end{equation}
where $\sigma_{ab}=\sigma_{ba}\in\mathfrak{S}_5$ is a transposition for mutually distinct indices $a,b\in\{1,\ldots,5\}$.
\begin{remark}
The group generated by $r_i$ ($i=1,\ldots,5$), $\sigma_{ab}$ ($a,b\in \{1,2,\ldots,5\}$), and $r_{34}$ constitute a symmetry group of the Garnier system in two variables. 
This group includes the affine Weyl group of type $B_5^{(1)}$ \cite{SuzukiT2005:MR2177118}. 
We note that the full symmetry group of the Garnier system has not yet been completely identified. 
\end{remark}

Based on $r_{34}$, for later convenience, we introduce the  transformations $r_{ab}=r_{ba}$ for other pairs of distinct subscripts  by 
\begin{equation}
r_{\sigma(a)\sigma(b)}=\sigma r_{ab}\sigma^{-1}\quad(\sigma\in\mathfrak{S}_5),
\end{equation}
where $a,b\in\{1,\ldots,5\}$, $a\not=b$. 
It is easy to see that we have  
\begin{equation}\label{fundamental_rel_r}
{r_{ab}}^2=1,\qquad r_ir_{ab}=r_{ab}r_i\quad(i\ne a,b). 
\end{equation}
The action of $r_{ab}$ on the parameters and the dependent variables is given by 
\begin{equation}
\begin{array}{lll}
r_{ab}&:&\theta_a\mapsto-\theta_a+1,\quad\theta_b\mapsto-\theta_b+1,
\quad\theta_c\mapsto-\theta_c\,(c\ne a,b),\\[1mm]
&&q_i\mapsto Q_i^{(ab)},\quad p_i\mapsto P_i^{(ab)},
\end{array}
\end{equation}
where $(\theta_3,\theta_4,\theta_5)=(\ka_0,\ka_1,\ka_{\infty})$, and
$Q_i^{(ab)},P_i^{(ab)}\,(i=1,2)$ are certain rational functions in the canonical variables $q_i,p_i\,(i=1,2)$. 

We also define the {\it Schlesinger transformations}
\begin{equation}
T_{ab}=r_ar_br_cr_dr_er_{ab},\quad{\mathcal T}_{ab}=r_ar_cr_dr_er_{ab}r_b,
\end{equation}
where $a,b,c,d,e\in\{1,\ldots,5\}$ are mutually distinct. These transformations commute with each other, and act on the parameters by
\begin{equation}
\begin{array}{llll}
T_{ab}&:&\theta_a\mapsto\theta_a+1,\qquad\theta_b\mapsto\theta_b+1,
&\quad\theta_c\mapsto\theta_c\quad(c\ne a,b),\\
{\mathcal T}_{ab}&:&\theta_a\mapsto\theta_a+1,\qquad\theta_b\mapsto\theta_b-1,
&\quad\theta_c\mapsto\theta_c\quad(c\ne a,b).\\
\end{array}
\end{equation}

\vskip2mm

\subsection{Transformations of the Hamiltonians}
One can introduce the $\tau$-variables through the Hamiltonians \cite{IKSY} in a similar manner to the Painlev\'e equations (see, for example, \cite{Okamoto:PVI, Noumi:book}). 
For this purpose, we here describe the action of $r_i$ ($i=1,\ldots,5$) and $r_{34}$ on the Hamiltonians. 

Following \cite{SuzukiT2005:MR2177118}, we first introduce the modified Hamiltonians $h_i^{(3)}$ ($i=1,2$) by adding correction terms, which do not depend on $p_i$, $q_i$, to the Hamiltonian $H_1$ \eqref{Garnier_hamiltonian} as
\begin{equation}
\begin{array}{ll}
h_1^{(3)}\!\!\!&
=H_1-\dfrac{\theta_1\theta_2}{2}\left(\dfrac{1}{t_1-1}-\dfrac{1}{t_1-t_2}\right)
+\dfrac{\theta_1(\ka_0-1)}{2}\left(\dfrac{1}{t_1}-\dfrac{1}{t_1-1}\right)
-\dfrac{\theta_1\ka_1}{2}\dfrac{1}{t_1-1}\\[4mm]
&+\dfrac{{\theta_1}^2}{8}
\left(\dfrac{1}{t_1}-\dfrac{3}{t_1-1}+\dfrac{1}{t_1-t_2}\right)
-\dfrac{{\theta_2}^2}{24}
\left(\dfrac{1}{t_1}+\dfrac{1}{t_1-1}-\dfrac{3}{t_1-t_2}\right)\\[4mm]
&+\dfrac{(\ka_0-1)^2}{24}
\left(\dfrac{3}{t_1}-\dfrac{1}{t_1-1}-\dfrac{1}{t_1-t_2}\right)
-\dfrac{{\ka_1}^2}{24}
\left(\dfrac{1}{t_1}+\dfrac{1}{t_1-1}+\dfrac{1}{t_1-t_2}\right)\\[4mm]
&-\dfrac{{\ka_{\infty}}^2}{24}
\left(\dfrac{1}{t_1}-\dfrac{3}{t_1-1}+\dfrac{1}{t_1-t_2}\right),
\end{array}
\end{equation}
and $h_2^{(3)}=\sigma_{12}(h_1^{(3)})$. Then we have
\begin{equation}
r_i(h_j^{(3)})=h_j^{(3)}\quad(i\ne3),\qquad
r_3(h_j^{(3)})=h_j^{(3)}-\dfrac{\ka_0q_j}{t_j^2(g_t-1)}
+\dfrac{\ka_0}{6}
\left(\dfrac{3}{t_1}-\dfrac{1}{t_1-1}-\dfrac{1}{t_1-t_2}\right).
\end{equation}
We see that the $1$-form
\begin{equation}
\omega_3=\sum_{j=1}^2h_j^{(3)}dt_j ,
\end{equation}
is invariant with respect to the action of $\sigma_{ab}\,(a,b\ne3)$, namely
\begin{equation}
\sigma_{ab}(\omega_3)=\omega_3\quad(a,b\ne3).
\end{equation}

Let us introduce the $1$-form $\omega_k\,(k=1,\ldots,5)$ and the Hamiltonians $h_j^{(k)}\,(j=1,2;\,k=1,\ldots,5)$ by
\begin{equation}
\sigma(\omega_k)=\omega_{\sigma(k)}\quad(\sigma\in\mathfrak{S}_5),
\end{equation}
and  
\begin{equation}\label{omega_k}
\omega_k=\sum_{j=1}^2h_j^{(k)}dt_j,
\end{equation}
respectively. 
We also define an auxiliary $1$-form $\omega$ by
\begin{equation}\label{omega_aux}
 \omega = \omega_5 + r_5(\omega_5).
\end{equation}
Then it follows that
\begin{equation}\label{r_i_on_omega}
\begin{array}{l}
r_i(\omega_k)=\omega_k, \quad(i\ne k),\\
r_{ab}(\omega_a)=\omega_b,\quad r_{ab}(\omega_b)=\omega_a.
\end{array}
\end{equation}

It is possible to verify by direct computation the following formulas among the above Hamiltonians
\begin{equation}\label{formula_1}
\begin{array}{l}
r_1(h_1^{(1)})+h_1^{(1)}-r_5(h_1^{(5)})-h_1^{(5)}
-\dfrac{\partial}{\partial t_1}\log q_1=
\dfrac{1}{3}\left(\dfrac{1}{t_1-t_2}-\dfrac{2}{t_1}\right),\\[4mm]
r_2(h_2^{(1)})+h_2^{(1)}-r_5(h_2^{(5)})-h_2^{(5)}
-\dfrac{\partial}{\partial t_2}\log q_1=
-\dfrac{1}{3}\left(\dfrac{1}{t_1-t_2}+\dfrac{1}{t_2-1}\right),\\[4mm]
r_3(h_j^{(3)})+h_j^{(3)}-r_5(h_j^{(5)})-h_j^{(5)}
-\dfrac{\partial}{\partial t_j}
\log\left(\dfrac{q_1}{t_1}+\dfrac{q_2}{t_2}-1\right)
=-\dfrac{1}{3}\left(\dfrac{1}{t_j-1}-\dfrac{1}{t_j}\right),\\[4mm]
r_4(h_j^{(4)})+h_j^{(4)}-r_5(h_j^{(5)})-h_j^{(5)}
-\dfrac{\partial}{\partial t_j}\log(q_1+q_2-1)
=-\dfrac{1}{3}\dfrac{1}{t_j-1},
\end{array}
\end{equation}
\begin{equation}\label{formula_2}
\begin{array}{l}
h_1^{(1)}+r_{34}(h_1^{(1)})-h_1^{(4)}-h_1^{(3)}
-\dfrac{\partial}{\partial t_1}\log(q_1p_1-\theta_1)\\[3mm]
\hskip20mm=-\dfrac{1}{4}
\left(\dfrac{1}{t_1}-\dfrac{1}{t_1-1}-\dfrac{1}{t_1-t_2}\right),\\[3mm]
h_2^{(1)}+r_{34}(h_2^{(1)})-h_2^{(4)}-h_2^{(3)}
-\dfrac{\partial}{\partial t_2}\log(q_1p_1-\theta_1)\\[3mm]
\hskip20mm=-\dfrac{1}{12}
\left(\dfrac{1}{t_2}+\dfrac{1}{t_2-1}+\dfrac{3}{t_1-t_2}\right),
\end{array}
\end{equation}
and 
\begin{equation}\label{formula_3}
\begin{array}{l}
r_5(h_j^{(5)})+r_{34}r_5(h_j^{(5)})-h_j^{(4)}-h_j^{(3)}
-\dfrac{\partial}{\partial t_j}\log(q_1p_1+q_2p_2+\al)\\[3mm]
\hskip20mm
=-\dfrac{1}{12}
\left(\dfrac{1}{t_j}-\dfrac{3}{t_j-1}+\dfrac{1}{t_j-t_k}\right),\\[4mm]
h_j^{(5)}+r_{34}(h_j^{(5)})-h_j^{(4)}-h_j^{(3)}
-\dfrac{\partial}{\partial t_j}\log(q_1p_1+q_2p_2+\al+\ka_{\infty})\\[3mm]
\hskip20mm
=-\dfrac{1}{12}
\left(\dfrac{1}{t_j}-\dfrac{3}{t_j-1}+\dfrac{1}{t_j-t_k}\right).
\end{array}
\end{equation}

From Equations \eqref{formula_1}, \eqref{formula_2} and \eqref{formula_3}, we can deduce the following relations among the $1$-forms $\omega_k$ $(k=1,\ldots, 5)$
\begin{equation}\label{formula'_1}
\begin{split}
&d\log q_1
=r_1(\omega_1)+\omega_1-r_5(\omega_5)-\omega_5
+\dfrac{2}{3}\dfrac{dt_1}{t_1}
+\dfrac{1}{3}\dfrac{dt_2}{t_2-1}
-\dfrac{1}{3}\dfrac{dt_1-dt_2}{t_1-t_2},\\[2mm]
&d\log q_2
=r_2(\omega_2)+\omega_2-r_5(\omega_5)-\omega_5
+\dfrac{1}{3}\dfrac{dt_1}{t_1-1}
+\dfrac{2}{3}\dfrac{dt_2}{t_2}
-\dfrac{1}{3}\dfrac{dt_1-dt_2}{t_1-t_2},\\[2mm]
&d\log\left(\dfrac{q_1}{t_1}+\dfrac{q_2}{t_2}-1\right)
=r_3(\omega_1)+\omega_3-r_5(\omega_5)-\omega_5
+\displaystyle\sum_{j=1}^2
\dfrac{1}{3}\left(\dfrac{1}{t_j-1}-\dfrac{1}{t_j}\right)dt_j,\\[2mm]
&d\log(q_1+q_2-1)
=r_4(\omega_1)+\omega_4-r_5(\omega_5)-\omega_5
+\displaystyle\sum_{j=1}^2\dfrac{1}{3}\dfrac{dt_j}{t_j-1},
\end{split}
\end{equation}
\begin{equation}\label{formula'_2}
\begin{split}
d\log(q_1p_1-\theta_1)
=&\omega_1+r_{34}(\omega_1)-\omega_4-\omega_3\\[1mm]
&+\dfrac{1}{4}\left(\dfrac{1}{t_1}-\dfrac{1}{t_1-1}\right)dt_1
+\dfrac{1}{12}\left(\dfrac{1}{t_2}+\dfrac{1}{t_2-1}\right)dt_2
-\dfrac{1}{4}\dfrac{dt_1-dt_2}{t_1-t_2},\\[3mm]
d\log(q_2p_2-\theta_2)
=&\omega_2+r_{34}(\omega_2)-\omega_4-\omega_3\\[1mm]
&+\dfrac{1}{12}\left(\dfrac{1}{t_1}+\dfrac{1}{t_1-1}\right)dt_1
+\dfrac{1}{4}\left(\dfrac{1}{t_2}-\dfrac{1}{t_2-1}\right)dt_2
-\dfrac{1}{4}\dfrac{dt_1-dt_2}{t_1-t_2}
\end{split}
\end{equation}
and 
\begin{equation}\label{formula'_3}
\begin{split}
&d\log(q_1p_1+q_2p_2+\al)
=\omega_5+r_{34}r_5(\omega_5)-\omega_4-\omega_3\\[1mm]
&\phantom{d\log(q_1p_1+q_2p_2+\al)}
+\displaystyle\sum_{j=1}^2\dfrac{1}{12}\left(\dfrac{1}{t_j}-\dfrac{3}{t_j-1}\right)
+\dfrac{1}{12}\dfrac{dt_1-dt_2}{t_1-t_2},\\[4mm]
&d\log(q_1p_1+q_2p_2+\al+\ka_{\infty})
=\omega_5+r_{34}(\omega_5)-\omega_4-\omega_3\\[1mm]
&\phantom{d\log(q_1p_1+q_2p_2+\al+\ka_{\infty})}
+\displaystyle\sum_{j=1}^2\dfrac{1}{12}\left(\dfrac{1}{t_j}-\dfrac{3}{t_j-1}\right)
+\dfrac{1}{12}\dfrac{dt_1-dt_2}{t_1-t_2},
\end{split}
\end{equation}
respectively.

\subsection{$\tau$-variables and bilinear relations}
Here we introduce the $\tau$-variables for the Garnier system in two variables and formulate the action of the following subgroup 
\begin{equation}\label{subgroup_G}
G=\br{r_i,\ r_{ab}},\qquad i=1,\ldots,5,\quad a,b\in\{1,\ldots,5\},
\end{equation}
on the $\tau$-variables; see also \cite{SuzukiT2005:MR2177118}. Note that the independent variables $t_i\,(i=1,2)$ are invariant under the action of this subgroup.

It is known that the $1$-forms $\omega_k$ ($k=1,\ldots,5$) given in Equations \eqref{omega_k} and
\eqref{omega_aux} are closed \cite{IKSY}, so that one can define the $\tau$-variables $\tau_k$
($k=1,\ldots,5$) and $\tau$ by
\begin{equation}\label{def_tau}
 \omega_k=d\log\tau_k \quad (k=1,\ldots,5) ,\quad \omega = d\log \tau,
\end{equation}
up to constant multiples. 

From Equations \eqref{r_i_on_omega}, \eqref{formula'_1} and \eqref{def_tau}, one can define the action of $r_i$ ($i=1,\ldots,5$) on the variables $\tau_k$ and the auxiliary variable $\tau$ by
\begin{equation}\label{r_on_tau}
\begin{array}{l}
r_i(\tau_k)=\tau_k\ (i\ne k),\quad r_i(\tau)=\tau
\  (i=1,\ldots,5),\\[2mm]
r_1(\tau_1)={t_1}^{-2/3}(t_2-1)^{-1/3}(t_1-t_2)^{1/3}
q_1\ \dfrac{\tau}{\tau_1},\\[4mm]
r_2(\tau_2)={t_2}^{-2/3}(t_1-1)^{-1/3}(t_1-t_2)^{1/3}
q_2\ \dfrac{\tau}{\tau_2},\\[3mm]
\displaystyle
r_3(\tau_3)=\prod_{j=1}^2{t_j}^{1/3}(t_j-1)^{-1/3}\,
\left(\dfrac{q_1}{t_1}+\dfrac{q_2}{t_2}-1\right)\dfrac{\tau}{\tau_3},\\[3mm]
\displaystyle
r_4(\tau_4)=\prod_{j=1}^2(t_j-1)^{-1/3}\,\big(q_1+q_2-1\big)
\dfrac{\tau}{\tau_4},\quad
r_5(\tau_5)=\dfrac{\tau}{\tau_5}.
\end{array}   
\end{equation}

The actions of $r_{ab}\,(a,b\in\{1,\ldots,5\})$ on the $\tau$-variables can be determined from Equations \eqref{formula'_2} and \eqref{formula'_3}. To this end, we define some preliminary notations. For mutually distinct indices $k,a,b\in\{1,\ldots,5\}$, we define $\varphi_{ab,k}=\varphi_{ba,k}$ by 
\begin{equation}\label{def_1:varphi}
\varphi_{34,k}=q_kp_k-\theta_k\quad(k=1,2),\qquad 
\varphi_{34,5}=-(q_1p_1+q_2p_2+\al+\ka_{\infty}),
\end{equation}
and 
\begin{equation}\label{def_2:varphi}
\sigma(\varphi_{ab,k})=\varphi_{\sigma(a)\sigma(b),\sigma(k)}
\quad(\sigma\in\mathfrak{S}_5).
\end{equation}
Further we introduce $\varphi_{a5}\,(a=1,2,3,4)$ by 
\begin{equation}
\begin{array}{ll}
\varphi_{15}=
\dfrac{\sigma_{13}\sigma_{45}r_{34}(g_1-1)}{g_1-1}
\varphi_{15,4}r_4(\varphi_{15,4}),&
\varphi_{25}=
\dfrac{\sigma_{23}\sigma_{45}r_{34}(g_1-1)}{g_1-1}
\varphi_{25,4}r_4(\varphi_{25,4}),\\[4mm]
\varphi_{35}=\dfrac{\sigma_{45}r_{34}(g_1-1)}{g_1-1}\ 
\varphi_{35,4}r_4(\varphi_{35,4}),&
\varphi_{45}=\dfrac{\sigma_{35}r_{34}(g_t-1)}{g_t-1}\ 
\varphi_{45,3}r_3(\varphi_{45,3}),
\end{array}
\end{equation}
where $g_1$ and $g_t$ are given in Equation \eqref{def_g}. Note that $\varphi_{ab,k}$ and $\varphi_{a5}$ are indeed polynomials of the canonical variables $q_i,p_i\,(i=1,2)$. Then the action of $r_{ab}$ on the $\tau$-variables can be expressed as follows 
\begin{equation}
\begin{array}{lll}
r_{ab}&:&\tau_a\mapsto\tau_b,\qquad\tau_b\mapsto\tau_a,\qquad
\tau_k\mapsto t_{ab,k}\varphi_{ab,k}\,
\dfrac{\tau_a\tau_b}{\tau_k}\quad(k\ne a,b),\\[2mm]
&&\tau\mapsto {t_{ab,5}}^2\,
\varphi_{ab,5}\,r_5(\varphi_{ab,5})
\dfrac{{\tau_a}^2{\tau_b}^2}{\tau}\quad(a,b\ne5),\\[2mm]
&&\tau\mapsto t_{a5}\,\varphi_{a5}\,\dfrac{{\tau_a}^2{\tau_5}^2}{\tau}, 
\end{array}
\end{equation}
where $t_{ab,k}=t_{ba,k}$ and $t_{ab}=t_{ba}$ are certain power products of $t_1,t_2,t_1-1,t_2-1$ and $t_1-t_2$. For readers' convenience we give the explicit form of them. Let introduce the notation 
\begin{equation}\label{chi}
 \chi(e_1,e_2,e_3,e_4,e_5)
 ={t_1}^{e_1}{t_2}^{e_2}(t_1-1)^{e_3}(t_2-1)^{e_4}(t_1-t_2)^{e_5},
\end{equation}
then $t_{ab,k}$ and $t_{ab}$ are given by 
\begin{equation}
\begin{array}{ll}
t_{12,3}=\chi(-\frac{1}{4},-\frac{1}{4},-\frac{1}{12},-\frac{1}{12}
,\frac{5}{12}),&
t_{12,4}=\chi(-\frac{1}{12},-\frac{1}{12},-\frac{1}{12},-\frac{1}{12}
,\frac{5}{12}),\\[1mm]
t_{12,5}=\chi(-\frac{1}{12},-\frac{1}{12},-\frac{1}{4},-\frac{1}{4}
,\frac{5}{12}),
\end{array}
\end{equation}
\begin{equation}
\begin{array}{ll}
t_{13,2}=\chi(\frac{5}{12},-\frac{1}{4},-\frac{1}{12},\frac{1}{4}
,-\frac{1}{4}),&
t_{13,4}=\sqrt{-1}\,\chi(\frac{5}{12},-\frac{1}{12},-\frac{1}{12},-\frac{1}{12}
,-\frac{1}{12}),\\[1mm]
t_{13,5}=\chi(\frac{5}{12},-\frac{1}{12},-\frac{1}{4},\frac{1}{4}
,-\frac{1}{12}),
\end{array}
\end{equation}
\begin{equation}
\begin{array}{ll}
t_{14,2}=-\chi(-\frac{1}{12},\frac{1}{4},-\frac{1}{12},\frac{1}{4}
,\frac{1}{4}),&
t_{14,3}=-\chi(-\frac{1}{4},\frac{1}{4},-\frac{1}{12},-\frac{1}{12}
,-\frac{1}{12}),\\[1mm]
t_{14,5}=-\chi(-\frac{1}{12},-\frac{1}{12},-\frac{1}{4},\frac{1}{4}
,-\frac{1}{12}),
\end{array}
\end{equation}
\begin{equation}
\begin{array}{ll}
t_{23,1}=-\chi(-\frac{1}{4},\frac{5}{12},\frac{1}{4},-\frac{1}{12}
,-\frac{1}{4}),&
t_{23,4}=-\chi(-\frac{1}{12},\frac{5}{12},-\frac{1}{12},-\frac{1}{12}
,-\frac{1}{12}),\\[1mm]
t_{23,5}=\sqrt{-1}\,\chi(-\frac{1}{12},\frac{5}{12},\frac{1}{4},-\frac{1}{4}
,-\frac{1}{12}),
\end{array}
\end{equation}
\begin{equation}
\begin{array}{ll}
t_{24,1}=-\sqrt{-1}\,\chi(\frac{1}{4},-\frac{1}{12},\frac{1}{4},-\frac{1}{12}
,-\frac{1}{4}),&
t_{24,3}=-\chi(\frac{1}{4},-\frac{1}{4},-\frac{1}{12},-\frac{1}{12}
,-\frac{1}{12}),\\[1mm]
t_{24,5}=-\chi(-\frac{1}{12},-\frac{1}{12},\frac{1}{4},-\frac{1}{4}
,-\frac{1}{12}),
\end{array}
\end{equation}
\begin{equation}
\begin{array}{ll}
t_{34,1}=\sqrt{-1}\,\chi(-\frac{1}{4},-\frac{1}{12},\frac{1}{4},-\frac{1}{12}
,\frac{1}{4}),&
t_{34,2}=\sqrt{-1}\,\chi(-\frac{1}{12},-\frac{1}{4},-\frac{1}{12},\frac{1}{4}
,\frac{1}{4}),\\[1mm]
t_{34,5}=\sqrt{-1}\,\chi(-\frac{1}{12},-\frac{1}{12},\frac{1}{4},\frac{1}{4}
,-\frac{1}{12}),
\end{array}
\end{equation}
\begin{equation}
\begin{array}{ll}
t_{15,2}=\chi(-\frac{1}{12},\frac{1}{4},\frac{5}{12},-\frac{1}{4}
,-\frac{1}{4}),&
t_{15,3}=\sqrt{-1}\,\chi(-\frac{1}{4},\frac{1}{4},\frac{5}{12},-\frac{1}{12}
,-\frac{1}{12}),\\[1mm]
t_{15,4}=\sqrt{-1}\,\chi(-\frac{1}{12},-\frac{1}{12},\frac{5}{12},-\frac{1}{12}
,-\frac{1}{12}),&
t_{15}=\chi(-\frac{1}{6},-\frac{1}{6},\frac{1}{2},\frac{1}{2}
,-\frac{1}{6}),
\end{array}
\end{equation}
\begin{equation}
\begin{array}{ll}
t_{25,1}=-\chi(\frac{1}{4},-\frac{1}{12},-\frac{1}{4},\frac{5}{12}
,-\frac{1}{4}),&
t_{25,3}=-\chi(\frac{1}{4},-\frac{1}{4},-\frac{1}{12},\frac{5}{12}
,-\frac{1}{12}),\\[1mm]
t_{25,4}=-\chi(-\frac{1}{12},-\frac{1}{12},-\frac{1}{12},\frac{5}{12}
,-\frac{1}{12}),&
t_{25}=-\chi(-\frac{1}{6},-\frac{1}{6},\frac{1}{2},\frac{1}{2}
,-\frac{1}{6}),
\end{array}
\end{equation}
\begin{equation}
\begin{array}{ll}
t_{35,1}=\chi(-\frac{1}{4},-\frac{1}{12},-\frac{1}{4},-\frac{1}{12}
,\frac{1}{4}),&
t_{35,2}=\chi(-\frac{1}{12},-\frac{1}{4},-\frac{1}{12},-\frac{1}{4}
,\frac{1}{4}),\\[1mm]
t_{35,4}=\chi(-\frac{1}{12},-\frac{1}{12},-\frac{1}{12},-\frac{1}{12}
,-\frac{1}{12}),&
t_{35}=\chi(-\frac{1}{6},-\frac{1}{6},\frac{1}{2},\frac{1}{2}
,-\frac{1}{6}),
\end{array}
\end{equation}
\begin{equation}
\begin{array}{ll}
t_{45,1}=-\chi(\frac{1}{4},-\frac{1}{12},-\frac{1}{4},-\frac{1}{12}
,\frac{1}{4}),&
t_{45,2}=-\chi(-\frac{1}{12},\frac{1}{4},-\frac{1}{12},-\frac{1}{4}
,\frac{1}{4}),\\[1mm]
t_{45,3}=-\sqrt{-1}\chi(\frac{1}{4},\frac{1}{4},-\frac{1}{12},-\frac{1}{12}
,-\frac{1}{12}),&
t_{45}=-\chi(-\frac{1}{6},-\frac{1}{6},\frac{1}{2},\frac{1}{2}
,-\frac{1}{6}).
\end{array}
\end{equation}

\medskip

On the level of the $\tau$-variables, the fundamental relations \eqref{fundamental_rel_r} are preserved except for some modifications. 
First we have
\begin{equation}
{r_{ab}}^2(\tau_k)=-\tau_k\quad(k\ne a,b).
\end{equation}
The commutativity of the Schlesinger transformations $T_{\mu},T_{\mu'}\in\big\{T_{ab},{\mathcal T}_{ab}\,\big|\,a,b\in\{1,\ldots,5\},\,a\ne b\big\}$ are modified by 
\begin{equation}
T_{\mu}T_{\mu'}(\tau_k)
=c_{\mu,\mu'}T_{\mu'}T_{\mu}(\tau_k)\quad(k=1,\ldots,5),\qquad
T_{\mu}T_{\mu'}(\tau)
={c_{\mu,\mu'}}^2T_{\mu'}T_{\mu}(\tau),
\end{equation}
where $c_{\mu,\mu'}$ are certain constants depending on the pair $(\mu, \mu')$.
\medskip

One can derive bilinear relations for the $\tau$-variables from the above formulation. 
It is easy to verify that we have
\begin{equation}
r_k(\varphi_{ab,k})-\varphi_{ab,k}=\theta_k.
\end{equation}
Substituting $t_{ab,k}\varphi_{ab,k}=\dfrac{\tau_kr_{ab}(\tau_k)}{\tau_a\tau_b}$ into the above formula, we get the bilinear relations 
\begin{equation}\label{bi:abk}
\tau_{k,k}r_{ab}(\tau_{k,k})-\tau_kr_{ab}(\tau_k)
=\theta_k\,t_{ab,k}\tau_a\tau_b ,
\end{equation}
for mutually distinct indices $a,b,k\in\{1,\ldots,5\}$, where $\tau_{k,k}=r_k(\tau_k)$. 
From the formulas \eqref{r_on_tau}, we get the bilinear relation 
\begin{align}\label{bi:constraint}
\displaystyle\prod_{j=1}^2{t_j}^{-1/3}(t_j-1)^{1/3}\,\tau_3\tau_{3,3}
=&{t_1}^{-1/3}(t_2-1)^{1/3}(t_1-t_2)^{-1/3}\,\tau_1\tau_{1,1}\notag\\
&+{t_2}^{-1/3}(t_1-1)^{1/3}(t_1-t_2)^{-1/3}\,\tau_2\tau_{2,2}-\tau_5\tau_{5,5}.
\end{align}

\section{Derivation of difference equations and constraints \label{Section:diffence equations}}
In this section, we derive a system of partial difference equations with some constraints from the bilinear relations for the $\tau$-variables of the Garnier system in two variables. 
The dependent variable of the system we derive can be regarded as defined on $\bbZ^4$. 
We obtain the difference equations \eqref{cr:3d} and \eqref{eqn_of_power_fnct} by restricting the domain of the system to $\bbZ^3$.

\subsection{$f$-variable and nonlinear relations}
Suggested by our preceding works  \cite{HKM:dSchwarzianKdV,AHKM:circle_patterns}, we introduce the $f$-variable in terms of the $\tau$-variables by 
\begin{equation}
f=\dfrac{\tau_{5,5}}{\tau_5}=\dfrac{r_5(\tau_5)}{\tau_5}.
\end{equation}
We then derive nonlinear relations that yield equations \eqref{cr:3d} and \eqref{eqn_of_power_fnct} in terms of the $f$-variable. 
Note that the action of the Schlesinger transformations $T_{ab},{\mathcal
T}_{ab}\,(a,b\in\{1,\ldots,5\})$ on $f$ commutes with each other, since the $f$-variable is defined by a ratio of $\tau$-variables. 
\begin{proposition}\label{prop:cr}
The $f$-variable satisfies the following relations 
\begin{equation}\label{cr_local}
\begin{array}{l}
\dfrac{f-{T_{14}}^{-1}(f)}{{T_{14}}^{-1}(f)-{T_{12}}^{-1}(f)}\,
\dfrac{{T_{12}}^{-1}(f)-{T_{24}}^{-1}(f)}{{T_{24}}^{-1}(f)-f}
=\dfrac{1-t_2}{1-t_1},\\[4mm]
\dfrac{f-{T_{24}}^{-1}(f)}{{T_{24}}^{-1}(f)-{T_{23}}^{-1}(f)}\,
\dfrac{{T_{23}}^{-1}(f)-{T_{34}}^{-1}(f)}{{T_{34}}^{-1}(f)-f}
=\dfrac{1}{1-t_2},\\[4mm]
\dfrac{f-{T_{34}}^{-1}(f)}{{T_{34}}^{-1}(f)-{T_{13}}^{-1}(f)}\,
\dfrac{{T_{13}}^{-1}(f)-{T_{14}}^{-1}(f)}{{T_{14}}^{-1}(f)-f}
=1-t_1.
\end{array}
\end{equation}
\end{proposition}

We remark that each equation in \eqref{cr_local} corresponds to the cross-ratio equation or discrete Schwarzian KdV equation \cite{NCWQ1984:MR763123,QNCL1984:MR761644} which is also known as
Q1$_{\de=0}$ \cite{ABS2003:MR1962121,ABS2009:MR2503862,BollR2011:MR2846098,BollR2012:MR3010833,BollR:thesis}.

\begin{proof} 
The bilinear relations \eqref{bi:abk} with $(a,b,k)=(1,4,5)$ and $(2,4,5)$ yield
\begin{equation}\label{eq':1-1}
f-{T_{14}}^{-1}(f)=-\ka_{\infty}\,
\textstyle\chi(-\frac{1}{12},-\frac{1}{12},-\frac{1}{4},\frac{1}{4},-\frac{1}{12})
\dfrac{\tau_1\tau_4}{\tau_5r_{14}(\tau_{5,5})},
\end{equation}
\begin{equation}\label{eq':1-2}
f-{T_{24}}^{-1}(f)=
-\ka_{\infty}\,
\textstyle\chi(-\frac{1}{12},-\frac{1}{12},\frac{1}{4},-\frac{1}{4},-\frac{1}{12})
\dfrac{\tau_2\tau_4}{\tau_5r_{24}(\tau_{5,5})},
\end{equation}
from which we have 
\begin{equation}\label{ratio:1-1}
\dfrac{f-{T_{14}}^{-1}(f)}{f-{T_{24}}^{-1}(f)}
=(t_1-1)^{-1/2}(t_2-1)^{1/2}
\dfrac{\tau_1r_{24}(\tau_{5,5})}{\tau_2r_{14}(\tau_{5,5})}.
\end{equation}
Applying the transformation ${T_{24}}^{-1}r_4$ to \eqref{eq':1-1}, we get 
\begin{equation}
{T_{24}}^{-1}(f)-{T_{12}}^{-1}(f)
=-\ka_{\infty}\,
\textstyle\chi(-\frac{1}{12},-\frac{1}{12},-\frac{1}{4},\frac{1}{4},-\frac{1}{12})
\dfrac{r_{24}(\tau_{1,1})\tau_2}{r_{24}(\tau_{5,5})r_{24}r_1r_{14}(\tau_5)}.
\end{equation}
Similarly, we obtain  
\begin{equation}
{T_{14}}^{-1}(f)-{T_{12}}^{-1}(f)
=-\ka_{\infty}\,
\textstyle\chi(-\frac{1}{12},-\frac{1}{12},\frac{1}{4},-\frac{1}{4},-\frac{1}{12})
\dfrac{r_{14}(\tau_{2,2})\tau_1}{r_{14}(\tau_{5,5})r_{14}r_2r_{24}(\tau_5)},
\end{equation}
by applying ${T_{14}}^{-1}r_4$ to \eqref{eq':1-2}. Due to the relation 
\begin{equation}
\dfrac{r_{24}(\tau_{1,1})}{r_{24}r_1r_{14}(\tau_5)}=
\dfrac{r_{14}(\tau_{2,2})}{r_{14}r_2r_{24}(\tau_5)},
\end{equation}
we have 
\begin{equation}\label{ratio:1-2}
\dfrac{{T_{24}}^{-1}(f)-{T_{12}}^{-1}(f)}{{T_{14}}^{-1}(f)-{T_{12}}^{-1}(f)}
=(t_1-1)^{-1/2}(t_2-1)^{1/2}
\dfrac{\tau_2r_{14}(\tau_{5,5})}{\tau_1r_{24}(\tau_{5,5})}.
\end{equation}
The formulas \eqref{ratio:1-1} and \eqref{ratio:1-2} yield the first relation of \eqref{cr_local}.

Similar computations give the other relations of Proposition \ref{prop:cr}. 
Although they are obtained by repetition of the same process, we exhibit the derivations in detail for readers' convenience. 
First, we show the third relation of \eqref{cr_local}. 
The bilinear relation \eqref{bi:abk} with $(a,b,k)=(3,4,5)$ gives
\begin{equation}\label{eq':2-2}
f-{T_{34}}^{-1}(f)=
\sqrt{-1}\,\ka_{\infty}\,
\textstyle\chi(-\frac{1}{12},-\frac{1}{12},\frac{1}{4},\frac{1}{4},-\frac{1}{12})
\dfrac{\tau_3\tau_4}{\tau_5r_{34}(\tau_{5,5})}.
\end{equation}
Then we have
\begin{equation}\label{ratio:2-1}
\dfrac{f-{T_{14}}^{-1}(f)}{f-{T_{34}}^{-1}(f)}
=\sqrt{-1}\,(t_1-1)^{-1/2}
\dfrac{\tau_1r_{34}(\tau_{5,5})}{\tau_3r_{14}(\tau_{5,5})},
\end{equation}
from \eqref{eq':1-1} and \eqref{eq':2-2}. 
Applying the transformation ${T_{34}}^{-1}r_4$ to \eqref{eq':1-1}, we get
\begin{equation}
{T_{34}}^{-1}(f)-{T_{13}}^{-1}(f)
=-\ka_{\infty}\,
\textstyle\chi(-\frac{1}{12},-\frac{1}{12},-\frac{1}{4},\frac{1}{4},-\frac{1}{12})
\dfrac{r_{34}(\tau_{1,1})\tau_3}{r_{34}(\tau_{5,5})r_{34}r_1r_{14}(\tau_5)}.
\end{equation}
Similarly we obtain 
\begin{equation}
{T_{14}}^{-1}(f)-{T_{13}}^{-1}(f)
=\sqrt{-1}\,\ka_{\infty}\,
\textstyle\chi(-\frac{1}{12},-\frac{1}{12},\frac{1}{4},\frac{1}{4},-\frac{1}{12})
\dfrac{r_{14}(\tau_{3,3})\tau_1}{r_{14}(\tau_{5,5})r_{14}r_3r_{34}(\tau_5)},
\end{equation}
by applying ${T_{14}}^{-1}r_4$ to \eqref{eq':2-2}. Noticing the relation 
\begin{equation}
\dfrac{r_{34}(\tau_{1,1})}{r_{34}r_1r_{14}(\tau_5)}=
\dfrac{r_{14}(\tau_{3,3})}{r_{14}r_3r_{34}(\tau_5)},
\end{equation}
we have 
\begin{equation}\label{ratio:2-2}
\dfrac{{T_{34}}^{-1}(f)-{T_{13}}^{-1}(f)}{{T_{14}}^{-1}(f)-{T_{13}}^{-1}(f)}
=\sqrt{-1}\,(t_1-1)^{-1/2}
\dfrac{\tau_3r_{14}(\tau_{5,5})}{\tau_1r_{34}(\tau_{5,5})}. 
\end{equation}
The formulas \eqref{ratio:2-1} and \eqref{ratio:2-2} yield the third relation of \eqref{cr_local}. 

Finally, we deduce the second relation of \eqref{cr_local}. 
We have from \eqref{eq':1-2} and \eqref{eq':2-2}
\begin{equation}\label{ratio:3-1}
\dfrac{f-{T_{24}}^{-1}(f)}{f-{T_{34}}^{-1}(f)}
=\sqrt{-1}\,(t_2-1)^{-1/2}
\dfrac{\tau_2r_{34}(\tau_{5,5})}{\tau_3r_{24}(\tau_{5,5})}.
\end{equation}
Applying the transformation ${T_{34}}^{-1}r_4$ to \eqref{eq':1-2}, we get
\begin{equation}
{T_{34}}^{-1}(f)-{T_{23}}^{-1}(f)
=-\ka_{\infty}\,
\textstyle\chi(-\frac{1}{12},-\frac{1}{12},\frac{1}{4},-\frac{1}{4},-\frac{1}{12})
\dfrac{r_{34}(\tau_{2,2})\tau_3}{r_{34}(\tau_{5,5})r_{34}r_2r_{24}(\tau_5)}. 
\end{equation}
Similarly we obtain 
\begin{equation}
{T_{24}}^{-1}(f)-{T_{23}}^{-1}(f)
=\sqrt{-1}\,\ka_{\infty}\,
\textstyle\chi(-\frac{1}{12},-\frac{1}{12},\frac{1}{4},\frac{1}{4},-\frac{1}{12})
\dfrac{r_{24}(\tau_{3,3})\tau_2}{r_{24}(\tau_{5,5})r_{24}r_3r_{34}(\tau_5)}
\end{equation}
by applying ${T_{24}}^{-1}r_4$ to \eqref{eq':2-2}. Then the relation 
\begin{equation}
\dfrac{r_{34}(\tau_{2,2})}{r_{34}r_2r_{24}(\tau_5)}=
\dfrac{r_{24}(\tau_{3,3})}{r_{24}r_3r_{34}(\tau_5)},
\end{equation}
gives
\begin{equation}\label{ratio:3-2}
\dfrac{{T_{34}}^{-1}(f)-{T_{23}}^{-1}(f)}{{T_{24}}^{-1}(f)-{T_{23}}^{-1}(f)}
=\sqrt{-1}\,(t_2-1)^{-1/2}
\dfrac{\tau_3r_{24}(\tau_{5,5})}{\tau_2r_{34}(\tau_{5,5})}. 
\end{equation}
The formuals \eqref{ratio:3-1} and \eqref{ratio:3-2} yield the second relation of \eqref{cr_local}, which completes the proof of Proposition \ref{prop:cr}.
\end{proof}


\begin{proposition}\label{Prop:H1}
The  $f$-variable satisfies the following relations 
\begin{equation}\label{H1_local}
\begin{array}{l}
\left(\dfrac{1}{f}-\dfrac{1}{{T_{14}}^{-1}(f)}\right)
\left({\mathcal T}_{54}(f)-{\mathcal T}_{51}(f)\right)
=\ka_{\infty}(\ka_{\infty}+1)
\,\chi(-\frac{1}{6},-\frac{1}{6},-\frac{1}{2},\frac{1}{2},-\frac{1}{6}),\\[5mm]
\left(\dfrac{1}{f}-\dfrac{1}{{T_{24}}^{-1}(f)}\right)
\left({\mathcal T}_{54}(f)-{\mathcal T}_{52}(f)\right)
=\ka_{\infty}(\ka_{\infty}+1)
\,\chi(-\frac{1}{6},-\frac{1}{6},\frac{1}{2},-\frac{1}{2},-\frac{1}{6}),\\[5mm]
\left(\dfrac{1}{f}-\dfrac{1}{{T_{34}}^{-1}(f)}\right)
\left({\mathcal T}_{54}(f)-{\mathcal T}_{53}(f)\right)
=-\ka_{\infty}(\ka_{\infty}+1)
\,\chi(-\frac{1}{6},-\frac{1}{6},\frac{1}{2},\frac{1}{2},-\frac{1}{6}), 
\end{array}
\end{equation}
where $\chi$ is given by \eqref{chi}.
\end{proposition}

We note that each of \eqref{H1_local} corresponds H1$_{\ep=0}$ in \cite{ABS2003:MR1962121,ABS2009:MR2503862,BollR2011:MR2846098,BollR2012:MR3010833,BollR:thesis}.
\begin{proof} 
We start with the relations 
\begin{equation}
\begin{array}{l}
f-{T_{14}}^{-1}(f)=-\ka_{\infty}\,
\textstyle\chi(-\frac{1}{12},-\frac{1}{12},-\frac{1}{4},\frac{1}{4},-\frac{1}{12})
\dfrac{\tau_1\tau_4}{\tau_5r_{14}(\tau_{5,5})},\\[4mm]
{\mathcal T}_{54}(f)-{\mathcal T}_{51}(f)=-(\ka_{\infty}+1)\,
\textstyle\chi(-\frac{1}{12},-\frac{1}{12},-\frac{1}{4},\frac{1}{4},-\frac{1}{12})
\dfrac{r_5r_{45}(\tau_{1,1})\tau_{5,5}}{\tau_4r_1r_5r_{45}r_{14}(\tau_{5,5})},
\end{array}
\end{equation}
the first of which is \eqref{eq':1-1} and the second can be obtained by applying $r_4T_{45}$ to the first. 
Due to the relation 
\begin{equation}
\dfrac{r_5r_{45}(\tau_{1,1})}{r_1r_5r_{45}r_{14}(\tau_{5,5})}
=-\dfrac{r_{14}(\tau_5)}{\tau_1},
\end{equation}
we have 
\begin{equation}
\begin{array}{l}
\left(f-{T_{14}}^{-1}(f)\right)
\left({\mathcal T}_{54}(f)-{\mathcal T}_{51}(f)\right)\\[1mm]
\hskip10mm
=-\ka_{\infty}(\ka_{\infty}+1)\,
\textstyle\chi(-\frac{1}{6},-\frac{1}{6},-\frac{1}{2},\frac{1}{2},-\frac{1}{6})
\,f\,{T_{14}}^{-1}(f),
\end{array}
\end{equation}
which is equivalent to the first relation of \eqref{H1_local}. 

The second and third relations can be also obtained in a similar way. 
Let us start with the relations 
\begin{equation}
\begin{array}{l}
f-{T_{24}}^{-1}(f)=
-\ka_{\infty}\,
\textstyle\chi(-\frac{1}{12},-\frac{1}{12},\frac{1}{4},-\frac{1}{4},-\frac{1}{12})
\dfrac{\tau_2\tau_4}{\tau_5r_{24}(\tau_{5,5})},\\[4mm]
{\mathcal T}_{54}(f)-{\mathcal T}_{52}(f)=
-(\ka_{\infty}+1)\,
\textstyle\chi(-\frac{1}{12},-\frac{1}{12},\frac{1}{4},-\frac{1}{4},-\frac{1}{12})
\dfrac{r_5r_{45}(\tau_{2,2})\tau_{5,5}}{\tau_4r_2r_5r_{45}r_{24}(\tau_{5,5})},
\end{array}
\end{equation}
the first of which is \eqref{eq':1-2} and the second can be obtained by applying $r_4T_{45}$ to the first. 
Due to the relation 
\begin{equation}
\dfrac{r_5r_{45}(\tau_{2,2})}{r_2r_5r_{45}r_{24}(\tau_{5,5})}
=-\dfrac{r_{24}(\tau_5)}{\tau_2},
\end{equation}
we have 
\begin{equation}
\begin{array}{l}
\left(f-{T_{24}}^{-1}(f)\right)
\left({\mathcal T}_{54}(f)-{\mathcal T}_{52}(f)\right)\\[1mm]
\hskip10mm
=-\ka_{\infty}(\ka_{\infty}+1)\,
\textstyle\chi(-\frac{1}{6},-\frac{1}{6},\frac{1}{2},-\frac{1}{2},-\frac{1}{6})
\,f\,{T_{24}}^{-1}(f),
\end{array}
\end{equation}
which is equivalent to the second relation of \eqref{H1_local}. 

Similarly we start with the relations 
\begin{equation}
\begin{array}{l}
f-{T_{34}}^{-1}(f)=
\sqrt{-1}\,\ka_{\infty}\,
\textstyle\chi(-\frac{1}{12},-\frac{1}{12},\frac{1}{4},\frac{1}{4},-\frac{1}{12})
\dfrac{\tau_3\tau_4}{\tau_5r_{34}(\tau_{5,5})},\\[4mm]
{\mathcal T}_{54}(f)-{\mathcal T}_{53}(f)=
\sqrt{-1}\,(\ka_{\infty}+1)\,
\textstyle\chi(-\frac{1}{12},-\frac{1}{12},\frac{1}{4},\frac{1}{4},-\frac{1}{12})
\dfrac{r_5r_{45}(\tau_{3,3})\tau_{5,5}}{\tau_4r_3r_5r_{45}r_{34}(\tau_{5,5})},
\end{array}
\end{equation}
the first of which is \eqref{eq':2-2} and the second can be obtained by applying $r_4T_{45}$ to to the first. 
Due to the relation 
\begin{equation}
\dfrac{r_5r_{45}(\tau_{3,3})}{r_3r_5r_{45}r_{34}(\tau_{5,5})}
=-\dfrac{r_{34}(\tau_5)}{\tau_3},
\end{equation}
we have 
\begin{equation}
\begin{array}{l}
\left(f-{T_{34}}^{-1}(f)\right)
\left({\mathcal T}_{54}(f)-{\mathcal T}_{53}(f)\right)\\[1mm]
\hskip15mm
=\ka_{\infty}(\ka_{\infty}+1)\,
\textstyle\chi(-\frac{1}{6},-\frac{1}{6},\frac{1}{2},\frac{1}{2},-\frac{1}{6})
\,f\,{T_{34}}^{-1}(f),
\end{array}
\end{equation}
which is equivalent to the third relation of \eqref{H1_local}. 
This completes the proof of Proposition \ref{Prop:H1}. 
\end{proof}

We next derive some relations that yield the similarity constraint \eqref{eqn_of_power_fnct}.
\begin{lemma}\label{Lemma:constraint}
The $f$-variable satisfies the following relations 
\begin{equation}\label{eqns_constraints_cocal}
\begin{array}{l}
\theta_1\,\dfrac{f-{\mathcal T}_{14}(f)}{T_{14}^{-1}(f)-{\mathcal T}_{14}(f)}+
\theta_2\,\dfrac{f-{\mathcal T}_{24}(f)}{T_{24}^{-1}(f)-{\mathcal T}_{24}(f)}+
\ka_0\,\dfrac{f-{\mathcal T}_{34}(f)}{T_{34}^{-1}(f)-{\mathcal T}_{34}(f)}\\[4mm]
\hskip60mm
=\dfrac{\theta_1+\theta_2+\ka_0+\ka_1+\ka_{\infty}-1}{2},\\[4mm]
\theta_1\,\dfrac{f-{\mathcal T}_{14}(f)}{T_{14}^{-1}(f)-{\mathcal T}_{14}(f)}\,T_{14}^{-1}(f)
+\theta_2\,\dfrac{f-{\mathcal T}_{24}(f)}{T_{24}^{-1}(f)-{\mathcal T}_{24}(f)}\,T_{24}^{-1}(f)
\\[4mm]
\hskip20mm
+\ka_0\,\dfrac{f-{\mathcal T}_{34}(f)}{T_{34}^{-1}(f)-{\mathcal T}_{34}(f)}\,T_{34}^{-1}(f)
=\dfrac{\theta_1+\theta_2+\ka_0+\ka_1-\ka_{\infty}-1}{2}\,f.
\end{array}
\end{equation}
\end{lemma}

\begin{proof}
The second relation can be obtained by applying the transformation $r_5$ to the first one. Then we demonstrate the proof of the first relation in detail. 
We start with the relations 
\begin{equation}
\begin{array}{l}
f-{T_{14}}^{-1}(f)=-\ka_{\infty}\,
\textstyle\chi(-\frac{1}{12},-\frac{1}{12},-\frac{1}{4},\frac{1}{4},-\frac{1}{12})
\dfrac{\tau_1\tau_4}{\tau_5r_{14}(\tau_{5,5})},\\[4mm]
f-{\mathcal T}_{14}(f)=-\ka_{\infty}\,
\textstyle\chi(-\frac{1}{12},-\frac{1}{12},-\frac{1}{4},\frac{1}{4},-\frac{1}{12})
\dfrac{\tau_{1,1}\tau_4}{\tau_5r_1r_{14}(\tau_{5,5})},
\end{array}
\end{equation}
the first of which is \eqref{eq':1-1} and the second can be obtained by applying $r_1$ to the first. 
Then we get
\begin{equation}\label{l1}
{T_{14}}^{-1}(f)-{\mathcal T}_{14}(f)
=\theta_1\,\ka_{\infty}\,\textstyle
\chi(\frac{1}{6},-\frac{1}{6},-\frac{1}{2},\frac{1}{6},\frac{1}{6})\,
\dfrac{{\tau_4}^2}{r_{14}(\tau_{5,5})r_1r_{14}(\tau_{5,5})},
\end{equation}
by using the bilinear relation \eqref{bi:abk} with $(a,b,k)=(4,5,1)$. 
Then we obtain 
\begin{equation}
\theta_1\,\dfrac{f-{\mathcal T}_{14}(f)}{T_{14}^{-1}(f)-{\mathcal T}_{14}(f)}=
-\textstyle\chi(-\frac{1}{4},\frac{1}{12},\frac{1}{4},\frac{1}{12},-\frac{1}{4})
\dfrac{\tau_{1,1}r_{14}(\tau_{5,5})}{\tau_4\tau_5}. 
\end{equation}

In a similar way one can express the other terms of the left-hand side of the first relation in Lemma \ref{Lemma:constraint} by the $\tau$-variables. 
For the second term we start with the relations
\begin{equation}
\begin{array}{l}
f-{T_{24}}^{-1}(f)=
-\ka_{\infty}
\textstyle\chi(-\frac{1}{12},-\frac{1}{12},\frac{1}{4},-\frac{1}{4},-\frac{1}{12})
\dfrac{\tau_2\tau_4}
      {\tau_5r_{24}(\tau_{5,5})},\\[4mm]
f-{\mathcal T}_{24}(f)=
-\ka_{\infty}
\textstyle\chi(-\frac{1}{12},-\frac{1}{12},\frac{1}{4},-\frac{1}{4},-\frac{1}{12})
\dfrac{\tau_{2,2}\tau_4}
      {\tau_5r_2r_{24}(\tau_{5,5})},
\end{array}
\end{equation}
the first of which is \eqref{eq':1-2} and the second can be obtained by applying $r_2$ to the first. 
We have
\begin{equation}\label{l2}
{T_{24}}^{-1}(f)-{\mathcal T}_{24}(f)
=\theta_2\,\ka_{\infty}\,\textstyle
\chi(-\frac{1}{6},\frac{1}{6},\frac{1}{6},-\frac{1}{2},\frac{1}{6})
\dfrac{{\tau_4}^2}{r_{24}(\tau_{5,5})r_2r_{24}(\tau_{5,5})}
\end{equation}
by using the bilinear relation \eqref{bi:abk} with $(a,b,k)=(4,5,2)$. 
Then we obtain 
\begin{equation}
\theta_2\,\dfrac{f-{\mathcal T}_{24}(f)}{T_{24}^{-1}(f)-{\mathcal T}_{24}(f)}=
-\textstyle\chi(\frac{1}{12},-\frac{1}{4},\frac{1}{12},\frac{1}{4},-\frac{1}{4})
\dfrac{\tau_{2,2}r_{24}(\tau_{5,5})}{\tau_4\tau_5}.
\end{equation}
Similarly, starting with the relations 
\begin{equation}
\begin{array}{l}
f-{T_{34}}^{-1}(f)=
\sqrt{-1}\,\ka_{\infty}\,
\textstyle\chi(-\frac{1}{12},-\frac{1}{12},\frac{1}{4},\frac{1}{4},-\frac{1}{12})
\dfrac{\tau_3\tau_4}
      {\tau_5r_{34}(\tau_{5,5})},\\[4mm]
f-{\mathcal T}_{34}(f)=
\sqrt{-1}\,\ka_{\infty}\,
\textstyle\chi(-\frac{1}{12},-\frac{1}{12},\frac{1}{4},\frac{1}{4},-\frac{1}{12})
\dfrac{\tau_{3,3}\tau_4}
      {\tau_5r_3r_{34}(\tau_{5,5})},
\end{array}
\end{equation}
the first of which is \eqref{eq':2-2} and the second can be obtained by applying $r_3$ to the first, we get
\begin{equation}\label{l3}
{T_{34}}^{-1}(f)-{\mathcal T}_{34}(f)
=\ka_0\,\ka_{\infty}\,\textstyle
\chi(\frac{1}{6},\frac{1}{6},\frac{1}{6},\frac{1}{6},-\frac{1}{6})\,
\dfrac{{\tau_4}^2}{r_{34}(\tau_{5,5})r_3r_{34}(\tau_{5,5})}. 
\end{equation}
Then the bilinear relation \eqref{bi:abk} with $(a,b,k)=(4,5,3)$ gives
\begin{equation}
\ka_0\,\dfrac{f-{\mathcal T}_{34}(f)}{T_{34}^{-1}(f)-{\mathcal T}_{34}(f)}=
\sqrt{-1}\,\textstyle\chi(-\frac{1}{4},-\frac{1}{4},\frac{1}{12},\frac{1}{12},\frac{1}{12})
\dfrac{\tau_{3,3}r_{34}(\tau_{5,5})}{\tau_4\tau_5}.
\end{equation}
Noticing that we have 
\begin{equation}\label{formula:rab_on_tau}
r_{14}(\tau_{5,5})=r_{45}(\tau_{1,1}),\quad
r_{24}(\tau_{5,5})=r_{45}(\tau_{2,2}),\quad
r_{34}(\tau_{5,5})=r_{45}(\tau_{3,3}),
\end{equation}
we find that the left-hand side of the first relation in Lemma \ref{Lemma:constraint} can be expressed by 
\begin{equation}
\begin{array}{l}
\dfrac{1}{t_{45,1}}\dfrac{\tau_{1,1}r_{45}(\tau_{1,1})}{\tau_4\tau_5}
+\dfrac{1}{t_{45,2}}\dfrac{\tau_{2,2}r_{45}(\tau_{2,2})}{\tau_4\tau_5}
+\dfrac{1}{t_{45,3}}\dfrac{\tau_{3,3}r_{45}(\tau_{3,3})}{\tau_4\tau_5}\\[4mm]
\hskip40mm
=r_1(\varphi_{45,1})+r_2(\varphi_{45,2})+r_3(\varphi_{45,3}). 
\end{array}
\end{equation}
From the definition \eqref{def_1:varphi} and \eqref{def_2:varphi} of $\varphi_{ab,k}$ we have 
\begin{equation}
r_1(\varphi_{45,1})+r_2(\varphi_{45,2})+r_3(\varphi_{45,3})
=-\al, 
\end{equation}
where $\al=-\dfrac{1}{2}(\theta_1+\theta_2+\ka_o+\ka_1+\ka_{\infty}-1)$. Then we obtain the first relation in Lemma \ref{Lemma:constraint}. 
\end{proof}

We have the following relation from Lemma \ref{Lemma:constraint}, which can be verified by a direct computation.

\begin{proposition}
The  $f$-variable satisfies the following relation 
\begin{equation}\label{similarity_local}
\begin{array}{l}
\ka_{\infty}f=
\theta_1\dfrac{({\mathcal T}_{14}(f)-f)(f-{T_{14}}^{-1}(f))}
      {{\mathcal T}_{14}(f)-{T_{14}}^{-1}(f)}\\[4mm]
\hskip10mm
+\theta_2\dfrac{({\mathcal T}_{24}(f)-f)(f-{T_{24}}^{-1}(f))}
      {{\mathcal T}_{24}(f)-{T_{24}}^{-1}(f)}
+\ka_0\dfrac{({\mathcal T}_{34}(f)-f)(f-{T_{34}}^{-1}(f))}
      {{\mathcal T}_{34}(f)-{T_{34}}^{-1}(f)}.
\end{array}
\end{equation}
\end{proposition}
We finally derive additional relations that will be used to derive the similarity constraint \eqref{eqn_of_power_fnct}.

\begin{proposition}
The $f$-variable satisfies the following relations 
\begin{equation}\label{constraint:2_local}
\begin{array}{l}
\widetilde{\chi}\,\ka_{\infty}\left(
\dfrac{\theta_1(t_2-1)}{(t_1-1)({T_{14}}^{-1}(f)-{\mathcal T}_{14}(f))}
+\dfrac{\theta_2}{{T_{24}}^{-1}(f)-{\mathcal T}_{24}(f)}\right.\\[4mm]
\hskip40mm\left.
-\dfrac{\ka_0(t_2-1)}{{T_{34}}^{-1}(f)-{\mathcal T}_{34}(f)}\right)
=\dfrac{1}{{T_{54}}^{-1}(f)},\\[4mm]
\widetilde{\chi}\,\ka_{\infty}\left(
\dfrac{\theta_1(t_2-1)}{(t_1-1)(1/{T_{14}}^{-1}(f)-1/{\mathcal T}_{14}(f))}\right.\\[4mm]
\hskip20mm\left.
+\dfrac{\theta_2}{1/{T_{24}}^{-1}(f)-1/{\mathcal T}_{24}(f)}
-\dfrac{\ka_0(t_2-1)}{1/{T_{34}}^{-1}(f)-1/{\mathcal T}_{34}(f)}\right)
={\mathcal T}_{54}(f), 
\end{array}
\end{equation}
where 
\begin{equation}\label{eqn:wide_chi}
\widetilde{\chi}=\textstyle\chi(-\frac{1}{6},-\frac{1}{6},\frac{1}{2},-\frac{1}{2},-\frac{1}{6}).
\end{equation}
\end{proposition}

\begin{proof}
The second relation can be obtained by applying the transformation $r_5$ to the first. Then we explain the proof on the first relation in detail. 
From \eqref{l1} and \eqref{formula:rab_on_tau} we have  
\begin{equation}
 \dfrac{\widetilde{\chi}\ka_{\infty}\theta_1}{{T_{14}}^{-1}(f)-{\mathcal T}_{14}(f)}=
{t_1}^{-1/3}(t_1-1)(t_2-1)^{-2/3}(t_1-t_2)^{-1/3}
r_{45}\left(\dfrac{\tau_1\tau_{1,1}}{\tau_5^2}\right).
\end{equation}
Similarly, we obtain from \eqref{l2}, \eqref{l3} and \eqref{formula:rab_on_tau} 
\begin{equation}
\begin{array}{l}
\dfrac{\widetilde{\chi}\ka_{\infty}\theta_2}{{T_{24}}^{-1}(f)-{\mathcal T}_{24}(f)}=
{t_2}^{-1/3}(t_1-1)^{1/3}(t_1-t_2)^{-1/3}
r_{45}\left(\dfrac{\tau_2\tau_{2,2}}{\tau_5^2}\right),\\[4mm]
\dfrac{\widetilde{\chi}\ka_{\infty}\ka_0}{{T_{34}}^{-1}(f)-{\mathcal T}_{34}(f)}
={t_1}^{-1/3}{t_2}^{-1/3}(t_1-1)^{1/3}(t_2-1)^{-2/3}
r_{45}\left(\dfrac{\tau_3\tau_{3,3}}{\tau_5^2}\right). 
\end{array}
\end{equation}
Then the bilinear relation \eqref{bi:constraint} gives 
\begin{equation}
\begin{array}{l}
\widetilde{\chi}\ka_{\infty}\left(
\dfrac{\theta_1(t_2-1)}{({T_{14}}^{-1}(f)-{\mathcal T}_{14}(f))(t_1-1)}
+\dfrac{\theta_2}{{T_{24}}^{-1}(f) - {\mathcal T}_{24}(f)}
-\dfrac{\ka_0(t_2-1)}{{T_{34}}^{-1}(f)-{\mathcal T}_{34}(f)}\right)\\[4mm]
\hskip10mm
=r_{45}\left(\dfrac{\tau_5\tau_{5,5}}{\tau_5^2}\right)
=r_{45}(f)=\dfrac{1}{{T_{54}}^{-1}(f)}, 
\end{array}
\end{equation}
which is the first relation of \eqref{constraint:2_local}. 
\end{proof}

\subsection{Restriction to a subsystem}\label{subsection:Restriction}
Here we introduce a subgroup of $G$ (defined in \eqref{subgroup_G}) and focus on the action of the subgroup. 
Considering some elements of infinite order of the subgroup we see that the relations for the $f$-variable obtained in the previous subsection yield a system of partial difference equations together with some constraints. 

Let us define the transformations $s_i^k\,(i=0,1\,;\,k=0,1,2,3)$ and $\pi_k\,(k=0,1,2,3)$ by 
\begin{equation}
\begin{array}{l}
s_1^k=r_k\quad(k=1,2,3),\qquad s_0^0=r_5,\\[1mm]
\pi_1=r_2r_3r_5r_{14},\quad\pi_2=r_1r_3r_5r_{24},\quad
\pi_3=r_1r_2r_5r_{34},\quad\pi_0=r_1r_2r_3r_{54},\\[1mm]
s_0^k=\pi_ks_1^k\pi_k\quad(k=1,2,3),\qquad s_1^0=\pi_0s_0^0\pi_0.
\end{array}
\end{equation}
Then these elements satisfy the following relations 
\begin{equation}
\begin{array}{l}
(s_i^k)^2=1,\qquad s_i^ks_j^l=s_j^ls_i^k\quad(k\ne l),\\[1mm]
(\pi_k)^2=1,\qquad\pi_k\pi_l=\pi_l\pi_k\quad(k\ne l),\\[1mm]
\pi_ks_i^l=s_i^l\pi_k\quad(k\ne l).
\end{array}
\end{equation}
This means that for each $k\in\{0,1,2,3\}$ the group generated by the transformations $s_0^k$ and $s_1^k$ is isomorphic to the affine Weyl group of type $A_1^{(1)}$ and $\pi_k$ plays the role of the automorphism of the Dynkin diagram; 
\begin{equation}
\br{s_0^k,s_1^k}\cong W(A_1^{(1)})\quad(k=0,1,2,3),\qquad
\br{s_0^k,s_1^k,\pi_k}\cong\widetilde{W}(A_1^{(1)})\quad(k=0,1,2,3).
\end{equation}
Also, these four groups $\br{s_0^k,s_1^k,\pi_k}\,(k=0,1,2,3)$ commute with each other. Let us introduce the parameters $\al_i^k\,(i=0,1\,;\,k=0,1,2,3,4)$ by 
\begin{equation}
\begin{array}{l}
\al_1^1=\theta_1,\quad\al_1^2=\theta_2,\quad
\al_1^3=\ka_0,\quad\al_1^4=\ka_1,\quad
\al_0^0=\ka_{\infty},\\[1mm]
\al_0^k=1-\al_1^k\quad(k=1,2,3,4),\qquad\al_1^0=1-\al_0^0.
\end{array}
\end{equation}
Then the action of the transformations $s_i^k$ and $\pi_k$ on these parameters is given by 
\begin{equation}\label{act_s:parameters}
\begin{array}{lll}
s_0^k&:&
\al_0^k\mapsto -\al_0^k,\quad
\al_1^k\mapsto\al_1^k+2\al_0^k,\\[1mm]
s_1^k&:&
\al_1^k\mapsto -\al_1^k,\quad
\al_0^k\mapsto\al_0^k+2\al_1^k
\end{array}
\end{equation}
and
\begin{equation}\label{act_pi:parameters}
\pi_k\,\,\,:\,\,\,(\al_0^k,\al_1^k,\al_0^4,\al_1^4)
\mapsto(\al_1^k,\al_0^k,\al_1^4,\al_0^4)\quad(k=0,1,2,3).
\end{equation}

Let us consider the elements 
\begin{equation}
\rho_k=\pi_ks_1^k\quad(k=0,1,2,3),
\end{equation}
of the group $\br{s_0^k,s_1^k,\pi_k}$. Each of them acts on the parameters $\al_0^k$ and $\al_1^k$ as a parallel translation: 
\begin{equation}
\rho_k\,\,\,:\,\,\,(\al_0^k,\al_1^k)\mapsto(\al_0^k+1,\al_1^k-1).
\end{equation}
Note that we have 
\begin{equation}
\begin{array}{ll}
\rho_k(f)={T_{k4}}^{-1}(f),&{\rho_k}^{-1}(f)={\mathcal T}_{k4}(f)
\qquad(k=1,2,3),\\[1mm]
\rho_0(f)={\mathcal T}_{54}(f),&{\rho_0}^{-1}(f)={T_{54}}^{-1}(f).
\end{array}
\end{equation}
Then the relations given in Proposition \ref{prop:cr} can be expressed by 
\begin{equation}\label{cr_by_rho}
\begin{array}{l}
\dfrac{f-\rho_1(f)}{\rho_1(f)-\rho_1\rho_2(f)}\,
\dfrac{\rho_1\rho_2(f)-\rho_2(f)}{\rho_2(f)-f}
=\dfrac{1-t_2}{1-t_1},\\[4mm]
\dfrac{f-\rho_2(f)}{\rho_2(f)-\rho_2\rho_3(f)}\,
\dfrac{\rho_2\rho_3(f)-\rho_3(f)}{\rho_3(f)-f}
=\dfrac{1}{1-t_2},\\[4mm]
\dfrac{f-\rho_3(f)}{\rho_3(f)-\rho_3\rho_1(f)}\,
\dfrac{\rho_3\rho_1(f)-\rho_1(f)}{\rho_1(f)-f}
=1-t_1. 
\end{array}
\end{equation}
Let us introduce an infinite family of variables $f_{l_1,l_2,l_3,l_0}\,(l_1,l_2,l_3,l_0\in\bbZ)$ by 
\begin{equation}
f_{l_1,l_2,l_3,l_0}={\rho_1}^{l_1}{\rho_2}^{l_2}{\rho_3}^{l_3}{\rho_0}^{l_0}(f)
\quad(l_1,l_2,l_3,l_0\in\bbZ). 
\end{equation}
By using the results in the previous subsection, we obtain several difference equations for the variables $f_{l_1,l_2,l_3,l_0}$. 
First, we have the following statement from \eqref{cr_by_rho} and \eqref{H1_local}. 

\begin{theorem}
The variables $f_{l_1,l_2,l_3,l_0}$ satisfy the partial difference equations 
\begin{equation}\label{cr:4d}
\begin{array}{l}
\dfrac{f_{l_1,l_2,l_3,l_0}-f_{l_1+1,l_2,l_3,l_0}}{f_{l_1+1,l_2,l_3,l_0}-f_{l_1+1,l_2+1,l_3,l_0}}\,
\dfrac{f_{l_1+1,l_2+1,l_3,l_0}-f_{l_1,l_2+1,l_3,l_0}}{f_{l_1,l_2+1,l_3,l_0}-f_{l_1,l_2,l_3,l_0}}
=\dfrac{1-t_2}{1-t_1},\\[4mm]
\dfrac{f_{l_1,l_2,l_3,l_0}-f_{l_1,l_2+1,l_3,l_0}}{f_{l_1,l_2+1,l_3,l_0}-f_{l_1,l_2+1,l_3+1,l_0}}\,
\dfrac{f_{l_1,l_2+1,l_3+1,l_0}-f_{l_1,l_2,l_3+1,l_0}}{f_{l_1,l_2,l_3+1,l_0}-f_{l_1,l_2,l_3,l_0}}
=\dfrac{1}{1-t_2},\\[4mm]
\dfrac{f_{l_1,l_2,l_3,l_0}-f_{l_1,l_2,l_3+1,l_0}}{f_{l_1,l_2,l_3+1,l_0}-f_{l_1+1,l_2,l_3+1,l_0}}\,
\dfrac{f_{l_1+1,l_2,l_3+1,l_0}-f_{l_1+1,l_2,l_3,l_0}}{f_{l_1+1,l_2,l_3,l_0}-f_{l_1,l_2,l_3,l_0}}
=1-t_1, 
\end{array}
\end{equation}
and 
\begin{equation}\label{H1}
\begin{array}{l}
\left(\dfrac{1}{f_{l_1,l_2,l_3,l_0}}-\dfrac{1}{f_{l_1+1,l_2,l_3,l_0}}\right)
\left(f_{l_1,l_2,l_3,l_0+1}-f_{l_1+1,l_2,l_3,l_0+1}\right)\\[2mm]
\hskip40mm
=(\al_0^0+l_0)(\al_0^0+l_0+1)
\chi(-\frac{1}{6},-\frac{1}{6},-\frac{1}{2},\frac{1}{2},-\frac{1}{6}),\\[2mm]
\left(\dfrac{1}{f_{l_1,l_2,l_3,l_0}}-\dfrac{1}{f_{l_1,l_2+1,l_3,l_0}}\right)
\left(f_{l_1,l_2,l_3,l_0+1}-f_{l_1,l_2+1,l_3,l_0+1}\right)\\[2mm]
\hskip40mm
=(\al_0^0+l_0)(\al_0^0+l_0+1)
\chi(-\frac{1}{6},-\frac{1}{6},\frac{1}{2},-\frac{1}{2},-\frac{1}{6}),\\[2mm]
\left(\dfrac{1}{f_{l_1,l_2,l_3,l_0}}-\dfrac{1}{f_{l_1,l_2,l_3+1,l_0}}\right)
\left(f_{l_1,l_2,l_3,l_0+1}-f_{l_1,l_2,l_3+1,l_0+1}\right)\\[2mm]
\hskip40mm
=-(\al_0^0+l_0)(\al_0^0+l_0+1)
\chi(-\frac{1}{6},-\frac{1}{6},\frac{1}{2},\frac{1}{2},-\frac{1}{6}). 
\end{array}
\end{equation}
\end{theorem}
We also get the following from \eqref{similarity_local} and \eqref{constraint:2_local}. 

\begin{theorem}
The variables $f_{l_1,l_2,l_3,l_0}$ satisfy the difference equations 
\begin{equation}\label{similarity}
\begin{array}{l}
(\al_0^0+l_0)\,f_{l_1,l_2,l_3,l_0}=
(\al_1^1-l_1)\dfrac{(f_{l_1-1,l_2,l_3,l_0}-f_{l_1,l_2,l_3,l_0})(f_{l_1,l_2,l_3,l_0}-f_{l_1+1,l_2,l_3,l_0})}
      {f_{l_1-1,l_2,l_3,l_0}-f_{l_1+1,l_2,l_3,l_0}}\\[4mm]
\hskip15mm
+(\al_1^2-l_2)\dfrac{(f_{l_1,l_2-1,l_3,l_0}-f_{l_1,l_2,l_3,l_0})(f_{l_1,l_2,l_3,l_0}-f_{l_1,l_2+1,l_3,l_0})}
      {f_{l_1,l_2-1,l_3,l_0}-f_{l_1,l_2+1,l_3,l_0}}\\[4mm]
\hskip15mm
+(\al_1^3-l_3)\dfrac{(f_{l_1,l_2,l_3-1,l_0}-f_{l_1,l_2,l_3,l_0})(f_{l_1,l_2,l_3,l_0}-f_{l_1,l_2,l_3+1,l_0})}
      {f_{l_1,l_2,l_3-1,l_0}-f_{l_1,l_2,l_3+1,l_0}}, 
\end{array}
\end{equation}
and 
\begin{equation}\label{constraint:2}
\begin{array}{l}
\dfrac{1}{f_{l_1,l_2,l_3,l_0-1}}=
\chi\,(\al_0^0+l_0)\left(
\dfrac{(\al_1^1-l_1)(t_2-1)}
{(t_1-1)(f_{l_1+1,l_2,l_3,l_0}-f_{l_1-1,l_2,l_3,l_0})}\right.\\[4mm]
\hskip30mm\left.
+\dfrac{(\al_1^2-l_2)}{f_{l_1,l_2+1,l_3,l_0}-f_{l_1,l_2-1,l_3,l_0}}
-\dfrac{(\al_1^3-l_3)(t_2-1)}{f_{l_1,l_2,l_3+1,l_0}-f_{l_1,l_2,l_3-1,l_0}}\right),\\[5mm]
f_{l_1,l_2,l_3,l_0+1}=
\chi\,(\al_0^0+l_0)\left(
\dfrac{(\al_1^1-l_1)(t_2-1)}
{(t_1-1)(1/f_{l_1+1,l_2,l_3,l_0}-1/f_{l_1-1,l_2,l_3,l_0})}\right.\\[4mm]
\hskip40mm
+\dfrac{(\al_1^2-l_2)}{1/f_{l_1,l_2+1,l_3,l_0}-1/f_{l_1,l_2-1,l_3,l_0}}\\[4mm]
\hskip50mm\left.
-\dfrac{(\al_1^3-l_3)(t_2-1)}{1/f_{l_1,l_2,l_3+1,l_0}-1/f_{l_1,l_2,l_3-1,l_0}}\right). 
\end{array}
\end{equation}
\end{theorem}
Equations \eqref{similarity} and \eqref{constraint:2} can be regarded as some constraints for the system of partial difference equations \eqref{cr:4d} and \eqref{H1}. 
One can determine all the values of $f_{l_1,l_2,l_3,l_0}$ from a finite number of initial values, as we will show in Appendix \ref{section:ABS_Garnier}. 

To conclude this section, we explicitly relate the above result to the defining equations of the discrete power function on a hexagonal lattice. 

\begin{theorem}
We introduce $f_{l_1,l_2,l_3}$ by 
\begin{equation}
f_{l_1,l_2,l_3}=f_{l_1,l_2,l_3,0}. 
\end{equation}
Then it satisfies the system of partial difference equations \eqref{cr:3d} with 
\begin{equation}
x_1=\dfrac{1-t_1}{1-t_2},\quad
x_2=1-t_2,\quad
x_3=\dfrac{1}{1-t_1}
\end{equation}
together with the similarity constraint \eqref{eqn_of_power_fnct}. 
\end{theorem}

When we consider the discrete power function on a hexagonal lattice, we have in equation \eqref{eqn_of_power_fnct} $\alpha_1^1 = \alpha_1^2 = \alpha_1^3 = 0$. 
In this case, by putting $l_1=l_2=l_3=0$, we see from \eqref{eqn_of_power_fnct} that $f_{0,0,0}$ cannot be freely chosen but instead must be $0$.
Also, by giving the initial values at the six points $f_{\pm1,0,0},f_{0,\pm1,0},f_{0,0,\pm1}$ one can determine all the values of $f_{l_1,l_2,l_3}$. 
However, the way to determination of values is different from the generic case.
If we start with the initial data at the three points $f_{1,0,0}, f_{0,1,0}, f_{0,0,-1}$, as given by \eqref{initial_conditions}, we can determine the values of $f_{l_1,l_2,l_3}$ consistently only for $l_1\ge 0, l_2\ge 0, l_3\le 0$. 
In this sense, the initial value problem for the discrete power function is consistent, but different from the generic case. 
The specialization $\alpha_1^1 = \alpha_1^2 = \alpha_1^3 = 0$ for the equation \eqref{eqn_of_power_fnct} suggests that the discrete power function on a hexagonal lattice relates to the special solution of the hypergeometric type to the Garnier system in two variables.

\section{Concluding remarks\label{ConcludingRemarks}} 
In this paper, we constructed a system of partial difference equations together with certain constraints from the discrete symmetry of the Garnier system in two variables. 
This contains the difference equations \eqref{cr:3d} and \eqref{eqn_of_power_fnct}, as a sub-system, which are associated with the discrete power functions on a hexagonal lattice. 

It is known that the Garnier system in two variables admits special solutions expressed by Appel's hypergeometric functions. 
Consequently, we expect to deduce an explicit formula of the discrete power function in terms of the hypergeometric $\tau$-functions of the Garnier system in two variables.  This will be reported in detail in a separate paper. 

\subsection*{Acknowledgment}
This research was supported by an Australian Research Council grant \#
DP160101728 and JSPS KAKENHI Grant Numbers JP16H03941,
JP19K14559 and JP17J00092.

\appendix
\section{A multi-dimensionally consistent system of partial difference equations}\label{section:ABS_Garnier}
The system of six multi-affine linear quad-equations is said to have a consistent-around-a-cube (CAC) property \cite{ABS2003:MR1962121,NW2001:MR1869690}, if they are placed on six faces of a cube in such a way that, given initial values on 4 vertices of the cube, the solutions of the equations are uniquely defined at each remaining vertex.
Integrable partial difference equations arise by iterating such consistent quad-equations on neighbouring cubes that tile the lattice $\bbZ^3$.

The concept of the CAC property can be extended to a hypercube.
The system of equations on a hypercube is said to be multi-dimensionally consistent, if the equations on every sub-cube of the hypercube have the CAC property.
Extension to $n$-dimensional hypercubes that are space-filling polytopes in the lattice $\bbZ^n$ leads to a system of integrable partial difference equations in $n$-directions.

In this section, we consider a system of partial difference equations \eqref{Q1_u} and \eqref{H1_u}, which is an extension of the system discussed in \cite{JKMNS2017:MR3741826} to four dimensions. 
This system is reduced to the system of difference equations \eqref{cr:4d} and \eqref{H1} under some conditions. 
The symmetry of the equations \eqref{Q1_u} and \eqref{H1_u} is taken over to the symmetry of the reduced equations. 
Furthermore we describe a birational action of the four Weyl groups $\br{s_0^k,s_1^k,\pi_k}\,(k=0,1,2,3)$, which recovers the system of difference equations  
\eqref{cr:3d} and \eqref{eqn_of_power_fnct}, through the system of equations \eqref{cr:4d} and \eqref{H1} together with the constraints \eqref{similarity} and \eqref{constraint:2}.

\subsection{A multi-dimensionally consistent system of partial difference equations}
Here, we consider a multi-dimensionally consistent system of partial difference equations and its symmetry. 

Let $u:\bbZ^4\mapsto\bbC$ be a function assigning a value to each vertex of the 4-dimensional integer lattice $\bbZ^4$. 
We denote the value of $u$ at the vertex $(l_1,l_2,l_3,l_0)\in\bbZ^4$ by $u_{l_1,l_2,l_3,l_0}$. Suppose $\mu^{(i)}_l\in\bbC\,(l\in\bbZ\,;\,i=0,1,2,3)$ are parameters assigned to each edge of the hypercube. 

We assume that the function $u$ satisfies the relations 
\begin{equation}\label{Q1_u}
\begin{array}{l}
\dfrac{(u_{l_1,l_2,l_3,l_0}+u_{l_1+1,l_2,l_3,l_0}) (u_{l_1,l_2+1,l_3,l_0}+u_{l_1+1,l_2+1,l_3,l_0})}
 {(u_{l_1,l_2,l_3,l_0}+u_{l_1,l_2+1,l_3,l_0}) (u_{l_1+1,l_2+1,l_3,l_0}+u_{l_1+1,l_2,l_3,l_0})}
 =\dfrac{\mu^{(1)}_{l_1}}{\mu^{(2)}_{l_2}},\\[4mm]
\dfrac{(u_{l_1,l_2,l_3,l_0}+u_{l_1,l_2+1,l_3,l_0}) (u_{l_1,l_2,l_3+1,l_0}+u_{l_1,l_2+1,l_3+1,l_0})}
 {(u_{l_1,l_2,l_3,l_0}+u_{l_1,l_2,l_3+1,l_0}) (u_{l_1,l_2+1,l_3+1,l_0}+u_{l_1,l_2+1,l_3,l_0})}
 =\dfrac{\mu^{(2)}_{l_2}}{\mu^{(3)}_{l_3}},\\[4mm]
\dfrac{(u_{l_1,l_2,l_3,l_0}+u_{l_1,l_2,l_3+1,l_0}) (u_{l_1+1,l_2,l_3,l_0}+u_{l_1+1,l_2,l_3+1,l_0})}
 {(u_{l_1,l_2,l_3,l_0}+u_{l_1+1,l_2,l_3,l_0}) (u_{l_1+1,l_2,l_3+1,l_0}+u_{l_1,l_2,l_3+1,l_0})}
 =\dfrac{\mu^{(3)}_{l_3}}{\mu^{(1)}_{l_1}},
\end{array}
\end{equation}
and 
\begin{equation}\label{H1_u}
\begin{array}{ll}
\left(\dfrac{1}{u_{l_1,l_2,l_3,l_0}}+\dfrac{1}{u_{l_1+1,l_2,l_3,l_0}}\right)(u_{l_1,l_2,l_3,l_0+1}+u_{l_1+1,l_2,l_3,l_0+1})
 =-\mu^{(1)}_{l_1}\mu^{(0)}_{l_0},\\[4mm]
\left(\dfrac{1}{u_{l_1,l_2,l_3,l_0}}+\dfrac{1}{u_{l_1,l_2+1,l_3,l_0}}\right)(u_{l_1,l_2,l_3,l_0+1}+u_{l_1,l_2+1,l_3,l_0+1})
 =-\mu^{(2)}_{l_2}\mu^{(0)}_{l_0},\\[4mm]
\left(\dfrac{1}{u_{l_1,l_2,l_3,l_0}}+\dfrac{1}{u_{l_1,l_2,l_3+1,l_0}}\right)(u_{l_1,l_2,l_3,l_0+1}+u_{l_1,l_2,l_3+1,l_0+1})
 =-\mu^{(3)}_{l_3}\mu^{(0)}_{l_0}.
\end{array}
\end{equation}
This list of six equations above corresponds to twenty four $2$-dimensional faces of a $4$-dimensional hypercube which forms a fundamental cell of $\bbZ^4$. 
See Figure \ref{fig:4-cube} for the corresponding labelling assigning coordinates to the vertices of a hypercube and parameters to each direction. 

\begin{figure}[t]
\begin{center}
\includegraphics[width=0.6\textwidth]{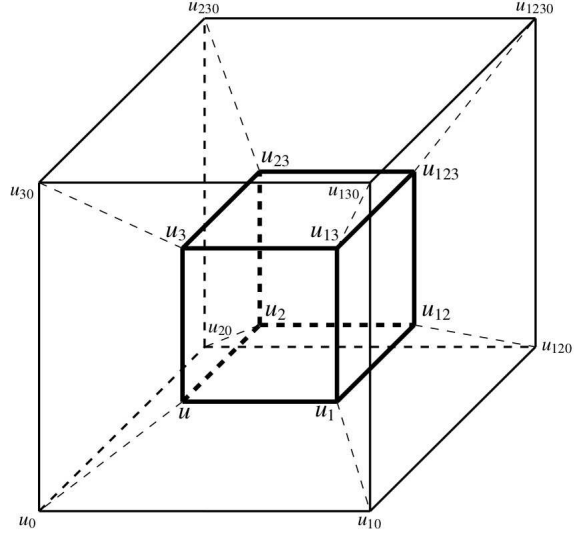}
\end{center}
\caption{
A $4$-dimensional hypercube.
The sixteen variables $u_{l_1,l_2,l_3,l_0}$,\dots, $u_{l_1+1,l_2+1,l_3+1,l_0+1}$ are assigned to the vertices of the hypercube.
Note that $u=u_{l_1,l_2,l_3,l_0}$ and the subscript $i$ 
means $+1$ shift in the $l_i$-direction, e.g. $u_1=u_{l_1+1,l_2,l_3,l_0}$ and $u_{10}=u_{l_1+1,l_2,l_3,l_0+1}$.
Moreover, each of the lattice parameters $\mu_{l_i}^{(i)}$, $i=0,1,2,3$, is assigned to the edge connecting $u$ and $u_i$ and to its parallel edges.
}
\label{fig:4-cube}
\end{figure}

Adler et al. \cite{ABS2003:MR1962121,ABS2009:MR2503862} and Boll \cite{BollR2011:MR2846098,BollR2012:MR3010833,BollR:thesis} classified the six face equations of 3D cube by using the CAC property and additional property called tetrahedron property. 
Every 3D sub-cube for the system of equations \eqref{Q1_u} and \eqref{H1_u} is fit with one in \cite{BollR2011:MR2846098,BollR2012:MR3010833,BollR:thesis}. 
In these papers, each of equations \eqref{Q1_u} is called Q1$_{\de=0}$ while each of equations \eqref{H1_u} is called H1$_{\ep=0}$. 
Note that each of partial difference equations \eqref{Q1_u} is known as the discrete Schwarzian KdV equation\cite{NCWQ1984:MR763123,QNCL1984:MR761644}. 

It is easy to see that the system of equations \eqref{Q1_u} and \eqref{H1_u} is invariant under the transformations $w_i^k~(i=0,1\,;\,k=0,1,2,3)$ and $\varpi_k~(k=0,1,2,3)$ given by 
\begin{equation}\label{eqn:lattic_symmetry_u}
\begin{array}{llll}
w_0^1&:&
u_{l_1,l_2,l_3,l_0}\mapsto u_{2-l_1,l_2,l_3,l_0},&
\mu_l^{(1)}\mapsto\mu_{1-l}^{(1)}\,,\\[1mm]
w_1^1&:&
u_{l_1,l_2,l_3,l_0}\mapsto u_{-l_1,l_2,l_3,l_0},&
\mu_l^{(1)}\mapsto\mu_{-1-l}^{(1)}\,,\\[1mm]
\varpi_1&:&
u_{l_1,l_2,l_3,l_0}\mapsto-u_{1-l_1,l_2,l_3,l_0},&
\mu_l^{(1)}\mapsto\mu_{-l}^{(1)}\,,\\[2mm]
w_0^2&:&
u_{l_1,l_2,l_3,l_0}\mapsto u_{l_1,2-l_2,l_3,l_0},&
\mu_l^{(2)}\mapsto\mu_{1-l}^{(2)}\,,\\[1mm]
w_1^2&:&
u_{l_1,l_2,l_3,l_0}\mapsto u_{l_1,-l_2,l_3,l_0},&
\mu_l^{(2)}\mapsto\mu_{-1-l}^{(2)}\,,\\[1mm]
\varpi_2&:&
u_{l_1,l_2,l_3,l_0}\mapsto-u_{l_1,1-l_2,l_3,l_0},&
\mu_l^{(2)}\mapsto\mu_{-l}^{(2)}\,,\\[2mm]
w_0^3&:&
u_{l_1,l_2,l_3,l_0}\mapsto u_{l_1,l_2,2-l_3,l_0},&
\mu_l^{(3)}\mapsto\mu_{1-l}^{(3)}\,,\\[1mm]
w_1^3&:&
u_{l_1,l_2,l_3,l_0}\mapsto u_{l_1,l_2,-l_3,l_0},&
\mu_l^{(3)}\mapsto\mu_{-1-l}^{(3)}\,,\\[1mm]
\varpi_3&:&
u_{l_1,l_2,l_3,l_0}\mapsto-u_{l_1,l_2,1-l_3,l_0},&
\mu_l^{(3)}\mapsto\mu_{-l}^{(3)},\\[2mm]
w_0^0&:&
u_{l_1,l_2,l_3,l_0}\mapsto\dfrac{1}{u_{l_1,l_2,l_3,-l_0}},&
\mu_l^{(0)}\mapsto\mu_{-1-l}^{(0)}\,,\\[4mm]
w_1^0&:&
u_{l_1,l_2,l_3,l_0}\mapsto\dfrac{1}{u_{l_1,l_2,l_3,-2-l_0}},&
\mu_l^{(0)}\mapsto\mu_{-3-l}^{(0)}\,,\\[4mm]
\varpi_0&:&
u_{l_1,l_2,l_3,l_0}\mapsto-\dfrac{1}{u_{l_1,l_2,l_3,-1-l_0}},&
\mu_l^{(0)}\mapsto\mu_{-2-l}^{(0)}\,. 
\end{array}
\end{equation}
These transformations can be read from 
the setting that the variables $u_{l_1,l_2,l_3,l_0}$ and the parameters $\mu_l^{(i)}$ are respectively assigned to the vetices and edges of the lattice. 
For each of $k\in\{0,1,2,3\}$ the transformations $w_0^k,w_1^k$ and $\varpi_k$ give a realization of the extended affine Weyl group $\widetilde{W}(A_1^{(1)})$, namely they satisfy the fundamental relations 
\begin{equation}
(w_0^k)^2=(w_1^1)^2=(\varpi_k)^2=1,\quad
\varpi_kw_0^k=w_1^k\varpi_k. 
\end{equation}
Furthermore, the four group $\br{w_0^k,w_1^k,\varpi_k}$, $k=0,1,2,3$, commute with each other. 

\subsection{Reduction of the system of equations \eqref{Q1_u} and \eqref{H1_u}}
Let us consider the reduction of \eqref{Q1_u} and \eqref{H1_u} to the system derived from the discrete symmetry of the Garnier system in two variables. 
To this end, we specialize $\mu_l^{(i)}\,(i=0,1,2,3)$ as  
\begin{equation}
\begin{array}{l}
\mu^{(1)}_l=\chi(-\frac{1}{6},-\frac{1}{6},-\frac{1}{2},\frac{1}{2},-\frac{1}{6}),\quad
\mu^{(2)}_l=\chi(-\frac{1}{6},-\frac{1}{6},\frac{1}{2},-\frac{1}{2},-\frac{1}{6}),\\[1mm]
\mu^{(3)}_l=-\chi(-\frac{1}{6},-\frac{1}{6},\frac{1}{2},\frac{1}{2},-\frac{1}{6}),\quad
\mu^{(0)}_l=(\al^0_0+l)\,(\al^0_0+l+1),
\end{array}
\end{equation}
where $\chi$ is given by \eqref{chi} and $\al_0^0$ is a complex parameter. 
Then, the equations \eqref{Q1_u} and \eqref{H1_u} become the equations \eqref{cr:4d} and \eqref{H1}, respectively, by the change of variables 
\begin{equation}
f_{l_1,l_2,l_3,l_0}=(-1)^{l_1+l_2+l_3}\,u_{l_1,l_2,l_3,l_0}.
\end{equation}
Recall that the difference equations \eqref{cr:4d} and \eqref{H1} are parts of the system derived from the discrete symmetry of the Garnier system in two variables. 
This system also contains the constraints \eqref{constraint:2} and 
\begin{equation}\label{constraint:1}
\begin{split}
&\dfrac{\al_1^1+\al_1^2+\al_1^3-\al_0^4+\al_0^0-l_1-l_2-l_3+l_0}{2}\\
&\quad=(\al_1^1-l_1)\,
\dfrac{f_{l_1,l_2,l_3,l_0}-f_{l_1-1,l_2,l_3,l_0}}{f_{l_1+1,l_2,l_3,l_0}-f_{l_1-1,l_2,l_3,l_0}}\\
&\qquad+(\al_1^2-l_2)\,
\dfrac{f_{l_1,l_2,l_3,l_0}-f_{l_1,l_2-1,l_3,l_0}}{f_{l_1,l_2+1,l_3,l_0}-f_{l_1,l_2-1,l_3,l_0}}\\
&\qquad+(\al_1^3-l_3)\,
\dfrac{f_{l_1,l_2,l_3,l_0}-f_{l_1,l_2,l_3-1,l_0}}{f_{l_1,l_2,l_3+1,l_0}-f_{l_1,l_2,l_3-1,l_0}},\\
&\dfrac{\al_1^1+\al_1^2+\al_1^3-\al_0^4-\al_0^0-l_1-l_2-l_3-l_0}{2}\,f_{l_1,l_2,l_3,l_0}\\
&\quad=(\al_1^1-l_1)\,
\dfrac{f_{l_1,l_2,l_3,l_0}-f_{l_1-1,l_2,l_3,l_0}}{f_{l_1+1,l_2,l_3,l_0}-f_{l_1-1,l_2,l_3,l_0}}\,
f_{l_1+1,l_2,l_3,l_0}\\
&\qquad+(\al_1^2-l_2)\,
\dfrac{f_{l_1,l_2,l_3,l_0}-f_{l_1,l_2-1,l_3,l_0}}{f_{l_1,l_2+1,l_3,l_0}-f_{l_1,l_2-1,l_3,l_0}}\,
f_{l_1,l_2+1,l_3,l_0}\\
&\qquad+(\al_1^3-l_3)\,
\dfrac{f_{l_1,l_2,l_3,l_0}-f_{l_1,l_2,l_3-1,l_0}}{f_{l_1,l_2,l_3+1,l_0}-f_{l_1,l_2,l_3-1,l_0}}\,
f_{l_1,l_2,l_3+1,l_0}, 
\end{split}
\end{equation}
where $\al_1^i$ $(i=1,2,3)$ and $\al_0^j$ $(j=0,4)$ are also complex parameters. 
Note that the constraints \eqref{constraint:1} directly come from Lemma \ref{Lemma:constraint}. 
By using the partial difference equations \eqref{Q1_u} and \eqref{H1_u} together with the constraints \eqref{constraint:2} and \eqref{constraint:1}, one can determine all the values of the variables $f_{l_1,l_2,l_3,l_0}$ from a finite number of initial values. 
For instance, one can choose 
\begin{equation}
 f_0=f_{0,0,0,0},\quad
 f_{\pm 1}=f_{\pm 1,0,0,0},\quad
 f_{\pm 2}=f_{0,\pm 1,0,0},\quad
 f_{\pm 3}=f_{0,0,\pm 1,0}
\end{equation}
as the initial values, where these seven values satisfy the relations 
\begin{equation}\label{eqn:initial_f}
\begin{split}
 &\dfrac{\al_1^1+\al_1^2+\al_1^3-\al_0^4+\al_0^0}{2}
 =\al_1^1\,\dfrac{f_0-f_{-1}}{f_{1}-f_{-1}}
 +\al_1^2\,\dfrac{f_0-f_{-2}}{f_{2}-f_{l-2}}
 +\al_1^3\,\dfrac{f_0-f_{-3}}{f_{3}-f_{-3}},\\
 &\dfrac{\al_1^1+\al_1^2+\al_1^3-\al_0^4-\al_0^0}{2}\,f_0
 =\al_1^1\,\dfrac{f_0-f_{-1}}{f_{1}-f_{-1}}\,f_{1}
 +\al_1^2\,\dfrac{f_0-f_{-2}}{f_{2}-f_{-2}}\,f_{2}
 +\al_1^3\,\dfrac{f_0-f_{-3}}{f_{3}-f_{-3}}\,f_{3}
\end{split}
\end{equation}
obtained from the constraints \eqref{constraint:1}. 

The symmetry of the system of equations \eqref{Q1_u} and \eqref{H1_u} described by $\br{w_0^k,w_1^k,\varpi_k}$, $k=0,1,2,3$, is taken over to that of the reduced system discussed in \S\ref{subsection:Restriction} by the following isomorphism:
\begin{equation}
w_1^k\mapsto s_1^k,\quad
w_0^k\mapsto s_0^k,\quad
\varpi_k\mapsto\pi_k\qquad(k=0,1,2,3).
\end{equation}
From the actions \eqref{eqn:lattic_symmetry_u} we obtain 
\begin{equation}
\begin{array}{lllllll}
s_0^1&:&
f_{l_1,l_2,l_3,l_0}\mapsto f_{2-l_1,l_2,l_3,l_0},&&
s_1^1&:&
f_{l_1,l_2,l_3,l_0}\mapsto f_{-l_1,l_2,l_3,l_0},\\[1mm]
\pi_1&:&f_{l_1,l_2,l_3,l_0}\mapsto f_{1-l_1,l_2,l_3,l_0},\\[2mm]
s_0^2&:&
f_{l_1,l_2,l_3,l_0}\mapsto f_{l_1,2-l_2,l_3,l_0},&&
s_1^2&:&f_{l_1,l_2,l_3,l_0}\mapsto f_{l_1,-l_2,l_3,l_0},\\[1mm]
\pi_2&:&f_{l_1,l_2,l_3,l_0}\mapsto f_{l_1,1-l_2,l_3,l_0},\\[2mm]
s_0^3&:&
f_{l_1,l_2,l_3,l_0}\mapsto f_{l_1,l_2,2-l_3,l_0},&&
s_1^3&:&f_{l_1,l_2,l_3,l_0}\mapsto f_{l_1,l_2,-l_3,l_0},\\[1mm]
\pi_3&:&f_{l_1,l_2,l_3,l_0}\mapsto f_{l_1,l_2,1-l_3,l_0},
\end{array}
\end{equation}
and 
\begin{equation}
\begin{array}{lllll}
s_0^0&:&f_{l_1,l_2,l_3,l_0}\mapsto\dfrac{1}{f_{l_1,l_2,l_3,-l_0}},&
\al_0^0\mapsto-\al_0^0,
\\[4mm]
s_1^0&:&f_{l_1,l_2,l_3,l_0}\mapsto\dfrac{1}{f_{l_1,l_2,l_3,-2-l_0}},&
\al_0^0\mapsto2-\al_0^0,
\\[4mm]
\pi_0&:&f_{l_1,l_2,l_3,l_0}\mapsto\dfrac{1}{f_{l_1,l_2,l_3,-1-l_0}},&
\al_0^0\mapsto 1-\al_0^0,
\end{array}
\end{equation}

Let us introduce the auxiliary parameters $\al_0^i$ $(i=1,2,3)$ and $\al_1^j$ $(j=0,4)$ by
\begin{equation}
\al_0^k=1-\al_1^k,\quad k=0,1,2,3,4.
\end{equation}
Since we impose the constraints \eqref{constraint:2} and \eqref{constraint:1}, one can formulate the birational action of $\br{s_0^k,s_1^k,\pi_k}$, $k=0,1,2,3$, on the field $\mathbb{K}(f_0,f_{\pm1},f_{\pm2},f_{\pm3})$ of rational functions with the coefficient field $\mathbb{K}=\bbC(\{\al_i^j\})$. 
Recall that the action on the coefficient field $\mathbb{K}$ is given by \eqref{act_s:parameters} and \eqref{act_pi:parameters}. 
It is easy to see that the transformations $s_1^k$ $(k=1,2,3)$ and $s_0^0$ are given by 
\begin{equation}
\begin{array}{lll}
s_1^k(f_k)=f_{-k},\quad s_1^k(f_{-k})=f_k&(k=1,2,3),\\[2mm]
s_0^0(f_k)=\dfrac{1}{f_k}&(k=0,\pm1,\pm2,\pm3). 
\end{array}
\end{equation}
The description of the transformations $\pi_k\,(k=1,2,3,0)$ is a little bit complicated. 
Regarding the transformation $\pi_1$, we have $\pi_1(f_0)=f_1,\pi_1(f_1)=f_0$ and 
\begin{equation}\label{act_of_pi1}
\begin{array}{l}
\dfrac{f_0-f_{-3}}{f_{-3}-\pi_1(f_{-3})}
\dfrac{\pi_1(f_{-3})-f_1}{f_1-f_0}=1-t_1,\quad
\dfrac{f_0-f_3}{f_3-\pi_1(f_3)}
\dfrac{\pi_1(f_3)-f_1}{f_1-f_0}=1-t_1,\\[5mm]
\dfrac{f_0-f_{-1}}{f_{-1}-\pi_1(f_{-2})}
\dfrac{\pi_1(f_{-2})-f_{-2}}{f_{-2}-f_0}=\dfrac{1-t_2}{1-t_1},\quad
\dfrac{f_0-f_1}{f_1-\pi_1(f_2)}
\dfrac{\pi_1(f_2)-f_2}{f_2-f_0}=\dfrac{1-t_2}{1-t_1}, 
\end{array}
\end{equation}
which come from the equations \eqref{Q1_u}. 
By solving each equations of \eqref{act_of_pi1} for $\pi_1(f_{-3})$, $\pi_1(f_3)$, $\pi_1(f_{-2})$ and $\pi_1(f_2)$, we see that these are expressed by the rational functions of $f_0,f_{\pm1},f_{\pm2},f_{\pm3}$. 
From the first relation of \eqref{eqn:initial_f} we also get
\begin{equation}
\dfrac{\al_0^1+\al_1^2+\al_1^3-\al_0^4+\al_0^0}{2}
=\al_0^1\,\dfrac{f_1-\pi_1(f_{-1})}{f_0-\pi_1(f_{-1})}
+\al_1^2\,\dfrac{f_1-\pi_1(f_{-2})}{f_2-\pi_1(f_{-2})}
+\al_1^3\,\dfrac{f_1-\pi_1(f_{-3})}{f_3-\pi_1(f_{-3})}, 
\end{equation}
which gives the expression of $\pi_1(f_{-1})$ by the rational functions of $f_0,f_{\pm1},f_{\pm2},f_{\pm3}$.
In a similar way, $\pi_k(f_i)$ $(k=2,3\,;\,i\in\{0,\pm1,\pm2,\pm3\})$ can be expressed in terms of the rational functions of $f_0,f_{\pm1},f_{\pm2},f_{\pm3}$. 
Let us consider the transformation $\pi_0$. 
From the first relation of \eqref{constraint:2}, we have 
\begin{equation}
\pi_0(f_0)=
\widetilde{\chi}\,\al_0^0\left(
\dfrac{\al_1^1(t_2-1)}{(t_1-1)(f_1-f_{-1})}
+\dfrac{\al_1^2}{f_2 - f_{-2}}
-\dfrac{\al_1^3(t_2-1)}{f_3-f_{-3}}\right),
\end{equation}
where $\widetilde{\chi}$ is given by \eqref{eqn:wide_chi}.
One can compute $\pi_0(f_k)$ and $\pi_0(f_{-k})$ for $k=1,2,3$ by $\pi_0(f_k)=\pi_0\pi_k(f_0)=\pi_k\pi_0(f_0)$ and $\pi_0(f_{-k})=\pi_0s_1^k(f_k)=s_1^k\pi_0(f_k)$, respectively. 
Then, we obtain the expression of $\pi_0(f_i)\,(i\in\{0,\pm1,\pm2,\pm3\})$ by the rational functions of $f_0,f_{\pm1},f_{\pm2},f_{\pm3}$. 
The transformations $s_0^k$ $(k=1,2,3)$ and $s_1^0$ are compositions of the above transformations, namely we have 
\begin{equation}
s_0^k=\pi_ks_1^k\pi_k\quad(k=1,2,3),\qquad
s_1^0=\pi_0s_0^0\pi_0. 
\end{equation}
We can verify that the action of $\br{s_0^k,s_1^k,\pi_k}$, $k=0,1,2,3$, constructed here is equivalent to that in Section \ref{subsection:Restriction}. 

\def\cprime{$'$} 

\end{document}